\documentclass[11pt]{amsart}
\usepackage{amsmath,amssymb,amscd,amsthm}
\usepackage{epsfig}

\newtheorem{theorem}{Theorem}
\newtheorem{lemma}{Lemma}
\newtheorem{proposition}{Proposition}
\newtheorem{remark}{Remark}

\newtheorem{definition}{Definition}
\newtheorem{corollary}{Corollary}
\newtheorem{assumption}{Assumption}

\newtheorem{claim}{Claim}
\newcommand{\f}[2]{\frac{#1}{#2}}

\newcommand{\al}{\alpha}

\newcommand{\de}{\delta}

\newcommand{\ve}{\varepsilon}

\newcommand{\la}{\lambda}

\newcommand{\si}{\sigma}

\newcommand{\om}{\omega}

\def\R{\mathbb{R}}

\def\C{\mathbb{C}}
\def\N{\mathbb{N}}

\newcommand{\cf}{\mathcal F}

\newcommand{\cn}{\mathcal N}

\newcommand{\ck}{\mathcal K}
\newcommand{\cl}{\mathcal L}
\newcommand{\ch}{\mathcal H}
\newcommand{\co}{\mathcal O}

\newcommand{\p}{\partial}

\newcommand{\beq}{\begin{equation}}
\newcommand{\eeq}{\end{equation}}
\newcommand{\beqna}{\begin{eqnarray*}}
\newcommand{\eeqna}{\end{eqnarray*}}
\newcommand{\beqn}{\begin{equation*}}
\newcommand{\eeqn}{\end{equation*}}
\newcommand{\bp}{\begin{proof}}
\newcommand{\ep}{\end{proof}}
\newcommand{\bprop}{\begin{proposition}}
\newcommand{\eprop}{\end{proposition}}
\newcommand{\bt}{\begin{theorem}}
\newcommand{\et}{\end{theorem}}
\newcommand{\bex}{\begin{Example}}
\newcommand{\eex}{\end{Example}}
\newcommand{\bc}{\begin{corollary}}
\newcommand{\ec}{\end{corollary}}
\newcommand{\bcl}{\begin{claim}}
\newcommand{\ecl}{\end{claim}}
\newcommand{\bl}{\begin{lemma}}
\newcommand{\el}{\end{lemma}}

{\begin{trivlist} \item[]{\em Proof }}%
{\hspace*{\fill}$\rule{.3\baselineskip}{.35\baselineskip}$\end{trivlist}}

\begin{document}

\title{Asymptotic stability of small gap solitons \\
in the nonlinear Dirac equations}

\author{Dmitry E. Pelinovsky}
\address{D.E. Pelinovsky
  Department of Mathematics and Statistics, McMaster
University, Hamilton, Ontario, Canada, L8S 4K1}

\author{Atanas Stefanov}
\address{A. Stefanov, Department of Mathematics, University of Kansas, 1460
Jayhawk Blvd, Lawrence, KS 66045--7523}
\thanks{D.E.P. is supported by NSERC. A.S. is supported in part by NSF-DMS \# 0908802.}

\maketitle

\begin{abstract}
We prove dispersive decay estimates for the one-dimensional Dirac
operator and use them to prove asymptotic stability of small gap
solitons in the nonlinear Dirac equations with quintic and
higher-order nonlinear terms.
\end{abstract}

\section{Introduction}

Asymptotic stability of solitary waves in the nonlinear Schr\"{o}dinger equation
has been considered in the space of three dimensions with a number of analytical techniques
\cite{PW,YT,SW,Gang,CM,KirrM}. Only recently, the asymptotic stability of solitary waves
was extended to the space of two dimensions \cite{M2,KirrZ} and one dimension
\cite{BS,Cuc,Mizum}.

Relatively little is known about the asymptotic stability of solitary waves in the nonlinear Dirac equations, which can be considered as a relativistic version of the nonlinear Schr\"{o}dinger equation.
Asymptotic stability of small bound states in the nonlinear Dirac equations in three
dimensions was constructed by Boussaid \cite{Boussaid}. Global existence and scattering to zero for
small initial data were obtained by Machihara {\em et al.} \cite{Machihara,Machihara2}, also in the space of three
dimensions.

We shall consider the asymptotic stability of solitary waves in
the nonlinear Dirac equations in one dimension. Since the energy
functional of the Dirac equations is sign-indefinite at the
linear wave spectrum, it is generally believed that the
solitary waves (referred to as gap solitons) must
be energetically (and nonlinearly) unstable. Indeed, gap solitons
are more disposed to spectral instabilities in the sense that
unstable eigenvalues may exist in a large subset of the existence domain \cite{ChPel-cme}. However, the
limit of small gap solitons corresponds to the nonrelativistic
limit, when the nonlinear Dirac equations can be reduced to the
nonlinear Schr\"{o}dinger equation \cite{Machihara2}. In this
limit, when the cubic nonlinear terms are considered, the gap solitons in one dimension
are typically stable both spectrally and orbitally. It is hence an interesting
question to study the nonlinear asymptotic stability of
the spectrally stable small gap solitons.

The spectral information is difficult in the case of the homogeneous nonlinear Dirac equation even in
the limit of small gap solitons. Isolated nonzero eigenvalues and resonances at the end points
of the continuous spectrum occur commonly in the problem \cite{Comech,ChPel-cme}.
To simplify the spectral information, we add a bounded
exponentially decaying potential to the one-dimensional nonlinear Dirac
equations and consider a local bifurcation of the small gap solitons from an isolated eigenvalue
of the self-adjoint Dirac operator. In this
way, our approach is similar to the one used by Mizumachi
\cite{Mizum} for the nonlinear Schr\"{o}dinger equation and by us
\cite{KPS} for the discrete nonlinear Schr\"{o}dinger equation
(see also \cite{CucTr} for similar results).

We shall avoid the dispersive decay estimates in weighted $L^2$
spaces, which are difficult for the nonlinear Dirac equations
(in contrast with the nonlinear Schr\"{o}dinger equations).
We shall instead derive the Strichartz estimates
directly from the Mizumachi estimates. The balance between
Strichartz and Mizumachi estimates allows us to control both the
nonlinear terms and the modulation equations for small gap
solitons and thus to prove their asymptotic stability for the
nonlinear Dirac equations with quintic and higher-order nonlinear
terms.

The article is organized as follows. Section 2 introduces the
nonlinear Dirac equations. Section 3 contains information about
the small gap solitons. Section 4 reports on
linearization and spectral stability for small gap solitons. Section 5
derives the modulation equations for parameters of gap solitons as
well as the time evolution equation for the dispersive remainder
term. Section 6 describes the spectral theory for the
one-dimensional Dirac operator. Section 7 deals with the linear
dispersive estimates for the semi-group associated with the Dirac
operator. Section 8 gives the proof of the main theorem.

We finish this section with the list of useful notations.

The inner product for complex-valued functions in $L^2(\R)$ is denoted by
\begin{equation}
\label{norm-l2}
\forall f,g \in L^2(\R) : \quad \langle f, g \rangle_{L^2} := \int_{\R} \bar{f}(x) g(x) dx.
\end{equation}

For any $f \in L^2(\R)$, we define the Fourier transform and its inverse by
\begin{eqnarray}
\label{Fourier}
\hat{f}(k) \equiv {\cf}(f) := \frac{1}{\sqrt{2\pi}}
\int_{-\infty}^\infty f(x) e^{-i x k} dx, \quad \check{f}(x)
\equiv {\cf}^{-1}(\hat{f}) := \frac{1}{\sqrt{2\pi}}
\int_{-\infty}^\infty \hat{f}(k) e^{i x k} dk.
\end{eqnarray}

Sobolev spaces are denoted by $W^{s,p}(\R)$ for $s \geq 0$ and $1
< p < \infty$ so that $H^s(\R) \equiv W^{s,2}(\R)$ and
$L^p(\R) \equiv W^{0,p}(\R)$. Beside Sobolev spaces, we will use
Strichartz spaces $L^p_t L^q_x$ and $L^q_x L^p_t$ defined
for $1 \leq p,q \leq \infty$ by the norms
\begin{eqnarray}
\label{Strichartz}
\|f \|_{L^p_t L^q_x} := \left( \int_0^T \| f(\cdot,t) \|_{L^q_x}^p dt \right)^{1/p}, \quad
\|f \|_{L^q_x L^p_t} := \left( \int_{\R} \| f(x,\cdot) \|_{L^p_t}^q dx \right)^{1/q},
\end{eqnarray}
where $T > 0$ is an arbitrary time including $T = \infty$.

Notation $\langle x \rangle = (1 + x^2)^{1/2}$ is used for the weights in $L^q_x$ norms.
The constant $C > 0$ is a generic constant, which may change from one line to another line.
A ball of radius $\delta > 0$ in function space $X$ centered at $0 \in X$ is denoted by $B_{\delta}(X)$.

Pauli matrices are defined by
$$
\sigma_1 = \left[ \begin{array}{cc} 0 & 1 \\ 1 & 0 \end{array} \right], \quad
\sigma_1 = \left[ \begin{array}{cc} 0 & i \\ -i & 0 \end{array} \right], \quad
\sigma_3 = \left[ \begin{array}{cc} 1 & 0 \\ 0 & -1 \end{array} \right].
$$
The $2$-by-$2$ identity matrix is denoted by $Id$.

Scalar functions are denoted by plain letters and vector functions with two components are denoted by
bold letters. For clarity of notations, we do not write the second arguments
for $W^{s,p}(\R)$, $H^s(\R)$, and $L^2(\R)$ when it is used for scalar or
vector functions.

\section{The nonlinear Dirac equations}

Consider the nonlinear Dirac equations
\begin{equation}
\label{cme}
\left\{ \begin{array}{cc} i (u_t + u_x) + v = \partial_{\bar{u}} W(u,v), \\
i (v_t - v_x) + u = \partial_{\bar{v}} W(u,v), \end{array} \right.
\end{equation}
where $(x,t) \in \mathbb{R}^2$, $(u,v) \in \mathbb{C}^2$, and $W(u,v) : \C^2 \to \R$ is a
nonlinear function which satisfies the following three conditions:
\begin{itemize}
\item symmetry $W(u,v) = W(v,u)$; \item gauge invariance $W(e^{i
\theta} u,e^{i \theta} v) = W(u,v)$ for any $\theta \in \R$; \item
polynomial in $(u,v)$ and $(\bar{u},\bar{v})$.
\end{itemize}

A general expansion of the nonlinear function $W(u,v)$ satisfying the
three properties above starts with quadratic and quartic terms
\begin{eqnarray}
W = \beta(x) (|u|^2 + |v|^2) + \gamma(x) (\bar{u} v
+ u \bar{v}) + W_N(u,v),
\end{eqnarray}
where $\beta(x), \gamma(x) : \R \to \R$ are bounded and
decaying potentials,
\begin{eqnarray}
W_N = \alpha_1 (|u|^4 + |v|^4) + \alpha_2 |u|^2 |v|^2 +
\alpha_3 (\bar{u} v + u \bar{v})^2 + \alpha_4 (|u|^2 + |v|^2)
(\bar{u} v + u \bar{v}) \label{potential}
\end{eqnarray}
is the nonlinear (quartic) potential, and $(\alpha_1,\alpha_2,\alpha_3,\alpha_4) \in
\R^4$ are numerical coefficients.

The standard example of the
nonlinear term occurs in the context of Bragg gratings, where
$\beta(x)$ and $\gamma(x)$ model optical defects in the periodic grating, whereas
\begin{equation}
\label{nonlinearity} W_N = \alpha (|u|^4 + 4 |u|^2 |v|^2 + |v|^4), \quad \alpha \in \R
\end{equation}
models the nonlinear coupling terms \cite{GWH,GSWK}.

Another example is
relevant to the massive Gross--Neveu model for spinors in relativity
theory \cite{Comech},
\begin{equation}
\label{nonlinearity-1} W_N = \alpha (\bar{u} v + u \bar{v})^2,
\quad \alpha \in \R.
\end{equation}

In other applications, $W_N$ may start with terms of the sixth and higher orders.
The following nonlinear potential is derived in the context of the Feshbach
resonance for Bose--Einstein condensates \cite{ChPorPel},
\begin{equation}
\label{nonlinearity-2} W_N = \alpha (|u|^2 + |v|^2) |u|^2 |v|^2,
\quad \alpha \in \R.
\end{equation}

Let us introduce the $2$-by-$2$ Dirac operator in one dimension
\begin{equation}
\label{Dirac-operator} \ch = \left[ \begin{array}{cc} -i
\partial_x + \beta(x) & \gamma(x) - 1\\ \gamma(x) - 1 & i \partial_x +
\beta(x) \end{array} \right] \equiv D + V(x),
\end{equation}
where
\begin{equation}
\label{a:30}
D = \left[ \begin{array}{cc} -i \partial_x & - 1 \\ - 1 & i \partial_x \end{array} \right], \quad
V(x) = \left[ \begin{array}{cc} \beta(x) & \gamma(x) \\
\gamma(x) & \beta(x) \end{array} \right].
\end{equation}

The nonlinear Dirac equations can be rewritten in the abstract
evolutionary form
\begin{equation}
\label{cme-formal} i \frac{d {\bf u}}{dt} = \ch {\bf u} + {\bf
N}({\bf u}), \quad {\bf N}({\bf u}) = \nabla_{\bar{\bf u}}
W_N(u,v),  \quad {\bf u} = \left[ \begin{array}{c} u \\ v
\end{array} \right], \quad \nabla_{\bar{\bf u}} = \left[ \begin{array}{c}
\partial_{\bar{u}} \\ \partial_{\bar{v}} \end{array} \right].
\end{equation}
where ${\bf N}({\bf u}) = {\co}(\| {\bf u} \|^3)$ as $\| {\bf u}
\| \to 0$ in any norm that forms Banach algebra (e.g. in $H^s(\R)$
for $s > \frac{1}{2}$). For the potentials (\ref{nonlinearity})
and (\ref{nonlinearity-1}), we have explicitly
$$
{\bf N}({\bf u}) = 2 \alpha \left[
\begin{array}{c} (|u|^2 + 2 |v|^2) u \\ (2 |u|^2 + |v|^2) v
\end{array} \right],\quad {\bf N}({\bf u}) = 2 \alpha
\left[
\begin{array}{c} |v|^2 u + v^2 \bar{u} \\ |u|^2 v + u^2 \bar{v}
\end{array} \right],
$$
For the technical reasons, these cubic nonlinear functions are not
sufficiently small when $(u,v)$ decays to zero. As a result, we
shall consider a more general class of the homogeneous polynomials
of $W_N$ (of even degree). Our arguments will be valid for the
quintic nonlinear functions which are generated from the
polynomial $W_N$ of degree six, e.g. from the function
(\ref{nonlinearity-2}).

Local existence of solutions in Sobolev space can be proved with
standard methods \cite{GWH}.

\begin{proposition}
\label{proposition-local-existence}
Let ${\bf u}_0 \in H^s(\R)$ for a fixed $s > \frac{1}{2}$ and
assume that $W$ satisfies the three conditions above. There exists
a $T > 0$ such that the nonlinear Dirac equations
(\ref{cme-formal}) admits a unique solution
$$
{\bf u}(t) \in C([0,T],H^s(\R)) \cap C^1([0,T],H^{s-1}(\R)),
$$
where ${\bf u}(t)$ depends continuously on the initial data ${\bf
u}(0) = {\bf u}_0$.
\end{proposition}

If the nonlinear functions $W_N(u,v)$ depends on $|u|^2$ and
$|v|^2$ only, e.g. for (\ref{nonlinearity}) and
(\ref{nonlinearity-2}), global well-posedness in $H^s(\R)$ with $s
\in \N$ can be proved \cite{GWH}. Little is known about the global
solutions even for small initial data for the general nonlinear
Dirac equations, e.g. for (\ref{nonlinearity-1}).

\section{Stationary small gap solitons}

Under the assumptions that $\beta(x), \gamma(x) \in L^{\infty}(\R)$,
Dirac operator $\ch$ is a densely defined, self-adjoint operator
in $L^2(\R)$ with the domain $H^1(\R)$. We shall further assume that
$$
\beta(x), \gamma(x) \to 0 \quad \mbox{\rm as} \quad |x| \to \infty
$$
at an exponential rate. The potentials represent a relatively compact perturbation
to the unbounded differential operator. By Weyl's Theorem,
the spectrum $\sigma(\ch) \subset \R$ contains the continuous spectrum at
$$
\sigma_c(\ch) \equiv (-\infty,-1] \cup [1,\infty).
$$

To simplify the construction of stationary small gap solitons, we
assume that $\ch$ admits only one simple isolated eigenvalue in the gap $(-1,1)$ of
the continuous spectrum $\sigma_c(\ch)$ and no resonances at $\pm 1$. Hence we add the assumption.

\begin{assumption}
Assume that
\begin{itemize}
\item $\beta, \gamma \in L^{\infty}(\R)$ and there is
$C > 0$ and $\kappa > 0$ such that
$$
|\beta(x)| + |\gamma(x)| \leq C e^{-\kappa |x|}, \quad x \in \R.
$$

\item $\sigma(\ch) \backslash \sigma_c(\ch) = \{ \omega_0 \}$,
where $\omega_0 \in (-1,1)$ is a simple eigenvalue of
$\ch$ with the $L^2$-normalized eigenfunction ${\bf u}_0 \in
H^1(\R)$.

\item No resonances occur at the end points $\pm 1$ of $\sigma_c(\ch)$.
\end{itemize}
\label{assumption-one}
\end{assumption}

Stationary gap solitons are given by
\begin{equation}
u(x,t) = U(x) e^{-i \omega t}, \quad v(x,t) = V(x) e^{-i \omega
t},
\end{equation}
where $\omega \in \mathbb{R}$ is a parameter and ${\bf U} =
[U,V]^T \in \mathbb{C}^2$ satisfies the system of differential
equations
\begin{equation}
\label{cme-ode} (\ch - \omega I) {\bf U} + {\bf N}({\bf U}) = {\bf
0}.
\end{equation}
If ${\bf U} \in H^1(\R)$, then ${\bf U} \in C(\R)$ and
${\bf U}(x) \to 0$ as $|x| \to \infty$ thanks to Sobolev's
embedding of $H^1(\R)$ to $C^0_b(\R)$. By Lemma 3.1 in \cite{ChPel-cme},
the stationary solution ${\bf U} \in H^1(\R)$ satisfies the symmetry
\begin{equation}
\label{symmetry-stationary-soliton}
U(x) = \bar{V}(x), \quad x \in \R.
\end{equation}

For the example of nonlinear function
(\ref{nonlinearity}) with $\alpha = \frac{1}{3}$,
when no linear potentials are present $\beta(x), \gamma(x) \equiv
0$, the stationary gap solitons are given in the explicit form for
any $\omega \in (-1,1)$
\begin{equation}
\label{gap-soliton} U(x) = \frac{\sqrt{1
- \omega^2}}{\sqrt{1 - \omega} \cosh(\sqrt{1 - \omega^2} x) + i
\sqrt{1 + \omega} \sinh(\sqrt{1 - \omega^2} x)} = \bar{V}(x).
\end{equation}
In particular, $\| {\bf U} \|_{L^{\infty}} \to 0$ as $\omega \to -1$,
which indicates the limit of small gap solitons.

As we explained in the introduction, the spectral information is
difficult in the case of homogeneous Dirac equations (without decaying potentials).
If $\beta(x)$ and $\gamma(x)$ are nonzero and Assumption \ref{assumption-one} is used,
the stationary gap solitons are not known in the explicit form but the local bifurcation
technique allows us to find a family of small gap solitons in a
one-sided neighborhood of $\omega = \omega_0$. To make it more
precise, let us assume that the nonlinear function is a
homogeneous polynomial in its variables.

\begin{assumption}
Assume that
$$
{\bf N}(a {\bf U}) = a^{2p+1} {\bf N}({\bf U}), \quad a \in \R,
$$
for a fixed integer $p \geq 1$. \label{assumption-two}
\end{assumption}

\begin{proposition}
Let Assumptions \ref{assumption-one} and \ref{assumption-two} be
true and
\begin{equation}
\label{condition-one}
\langle {\bf u}_0, {\bf N}({\bf u}_0) \rangle_{L^2} > 0.
\end{equation}
For sufficiently small $\epsilon > 0$, there is a family of
solutions ${\bf U} \in H^1(\R)$ of system (\ref{cme-ode}) for any
$\omega \in (\omega_0, \omega + \epsilon)$ such that the map
$(\omega_0, \omega + \epsilon) \ni \omega \mapsto {\bf U} \in H^1(\R)$ is $C^1$ and
\begin{equation}
\label{bound-one}
\| {\bf U} - a {\bf u}_0 \|_{H^1} = {\co}(a^{2p+1}), \quad
|\omega - \omega_0 | = {\co}(a^{2p}), \quad \mbox{\rm as} \quad
a \to 0.
\end{equation}
\label{proposition-local-bifurcation}
\end{proposition}

\begin{proof}
Thanks to Assumption \ref{assumption-one}, we use the decomposition
$$
{\bf U} = a {\bf u}_0 + {\bf V}, \quad a \in \R, \quad \langle {\bf u}_0, {\bf V} \rangle_{L^2} = 0.
$$

Let $P_0 : L^2(\R) \to {\rm Ran}(\ch - \omega_0 I) \subset L^2(\R)$ be the orthogonal projection operator
so that ${\bf V} = P_0 {\bf V} \in {\rm Ran}(\ch - \omega_0 I)$.
The stationary equation (\ref{cme-ode}) becomes the following system of two equations
\begin{eqnarray*}
\left\{ \begin{array}{l} P_0 (\ch - \omega I) P_0 {\bf V} + P_0 {\bf N}(a {\bf u}_0 + {\bf V})={\bf
0}, \\
(\omega_0 - \omega) a + \langle {\bf u}_0, {\bf N}(a {\bf u}_0 + {\bf V}) \rangle_{L^2} = 0.
\end{array} \right.
\end{eqnarray*}

Operator $P_0 (\ch - \omega_0 I) P_0 : H^1(\R) \to L^2(\R)$ is invertible for $\omega$ near $\omega_0$,
whereas ${\bf N}({\bf U})$ is a $C^{2p+1}$ function near ${\bf 0} \in H^1(\R)$.
By the Implicit Function Theorem, there is a unique $C^{2p+1}$ map $\R \ni a \mapsto {\bf V} \in H^1(\R) \cap {\rm Ran}(L_0 - \omega_0 I)$ such that ${\bf V}$ satisfies the first equation of the system and
there are $a_0 > 0$ and $C > 0$ such that for all $a \in (-a_0,a_0)$,
\begin{equation}
\label{bound-on-V}
\| {\bf V} - a^{2p+1} P_0 (\ch - \omega_0 I)^{-1} P_0 {\bf N}({\bf u}_0) \|_{H^1} \leq C a^{4p+1}.
\end{equation}

Let us substitute the map $\R \ni a \mapsto {\bf V} \in H^1(\R) \cap {\rm Ran}(L_0 - \omega_0 I)$ to the scalar equation
$$
F(a,\omega) = (\omega_0 - \omega) + a^{-1} \langle {\bf u}_0, {\bf N}(a {\bf u}_0 + {\bf V}) \rangle_{L^2} = 0.
$$
Thanks to the bound (\ref{bound-on-V}) and Assumption \ref{assumption-two}, there are $a_0 > 0$ and $C > 0$ such that for all $a \in (-a_0,a_0)$, there is only one solution of $F(a,\omega) = 0$ for $\omega = \omega(a)$ satisfying the bound
\begin{equation}
\label{bound-on-omega}
|\omega_0 + a^{2p}  \langle {\bf u}_0, {\bf N}({\bf u}_0) \rangle_{L^2} - \omega | \leq C a^{4p}.
\end{equation}
Under condition (\ref{condition-one}), we have $\omega > \omega_0$ and the bounds (\ref{bound-one})
follow from (\ref{bound-on-V}) and (\ref{bound-on-omega}).
\end{proof}

\begin{remark}
Proposition \ref{proposition-local-bifurcation} is valid if
\begin{equation}
\label{reverse-inequality}
\langle {\bf u}_0, {\bf N}({\bf u}_0) \rangle_{L^2} < 0,
\end{equation}
but the family of solutions  ${\bf U} \in H^1(\R)$ of system
(\ref{cme-ode}) exist for $\omega \in (\omega_0-\epsilon,
\omega)$ under the condition (\ref{reverse-inequality}).
\end{remark}

\section{Linearization and spectral stability}

Linearization is performed after writing
\begin{equation}
\label{linearization}
\left\{ \begin{array}{cc}
u(x,t) = e^{-i \omega t} \left[ U(x) + U_1(x) e^{\lambda t} + \bar{U}_2(x) e^{\bar{\lambda} t} \right], \\
v(x,t) = e^{-i \omega t} \left[ V(x) + V_1(x) e^{\lambda t} + \bar{V}_2(x) e^{\bar{\lambda} t} \right],
\end{array} \right.
\end{equation}
and neglecting quadratic terms with respect to the vectors
$$
{\bf U}_1 = [U_1,V_1]^T \in \C^2, \quad
{\bf U}_2 = [U_2,V_2]^T \in \mathbb{C}^2.
$$

From the derivatives of ${\bf N}({\bf U})$, we obtain
the expansion in any norm that forms a Banach algebra with respect to
the pointwise multiplication,
\begin{equation}
\label{expansion-N}
{\bf N}({\bf U} + {\bf U}_1) = {\bf N}({\bf U}) + V_{11} {\bf U}_1 + V_{12} \bar{\bf U}_1 +
{\co}(\| {\bf U}_1 \|^2),
\end{equation}
where $V_{11}$ and $V_{12}$ are $2 \times 2$ matrices with exponentially decaying
coefficients, which are given explicitly by
\begin{eqnarray}
\label{potentials-explicit-linearized}
V_{11} = \left[ \begin{array}{cc} \partial_{\bar{U} U}^2 W_N & \partial_{\bar{U} V}^2 W_N \\
\partial_{\bar{V} U}^2 W_N & \partial_{\bar{V} V}^2 W_N \end{array} \right] = \bar{V}_{11}^T,
\quad
V_{12} = \left[ \begin{array}{cc} \partial_{\bar{U} \bar{U}}^2 W_N & \partial_{\bar{U} \bar{V}}^2 W_N \\
\partial_{\bar{V} \bar{U}}^2 W_N & \partial_{\bar{V} \bar{V}}^2 W_N \end{array} \right] = V_{12}^T.
\end{eqnarray}

Substituting (\ref{linearization}) to the nonlinear Dirac
equations (\ref{cme-formal}) and using expansion
(\ref{expansion-N}), we obtain the linear eigenvalue problem
\begin{equation}
\label{linearized-operator}
\left\{ \begin{array}{cc} i \lambda {\bf U}_1 = (\ch - \omega I) {\bf U}_1 + V_{11} {\bf U}_1 + V_{12} {\bf U}_2, \\
-i \lambda {\bf U}_2 = (\bar{H}_0 - \omega I) {\bf U}_2 + \bar{V}_{12} {\bf U}_1 + \bar{V}_{11} {\bf U}_2. \end{array} \right.
\end{equation}

We should distinguish the self-adjoint operator $H_{\omega} : H^1(\R) \to L^2(\R)$ given by
$$
H_{\omega} = \left[ \begin{array}{cc} \ch - \omega I & 0 \\ 0 & \bar{H}_0 - \omega I \end{array} \right] +
\left[ \begin{array}{cc} V_{11}  & V_{12} \\ \bar{V}_{12} & \bar{V}_{11} \end{array} \right]
$$
and the non-self-adjoint linearization operator $L_{\omega} = - i
\sigma H_{\omega} : H^1(\R) \to L^2(\R)$, where
\begin{equation}
\label{matrix-sigma} \sigma = \left[ \begin{array}{cc} I & 0 \\ 0
& -I \end{array} \right].
\end{equation}
Both operators act on $[{\bf U}_1,{\bf U}_2]^T$.

Symmetry (\ref{symmetry-stationary-soliton}) imply that
$$
\partial_{\bar{U} U}^2 W_N = \partial_{\bar{V} V}^2 W_N, \quad
\partial_{\bar{U} \bar{U}}^2 W_N = \partial_{V V}^2 W_N, \quad
\partial_{U V}^2 W_N = \partial_{\bar{U} \bar{V}}^2 W_N.
$$
By Theorem 4.1 in \cite{ChPel-cme}, the self-adjoint operator $H_{\omega}$ and
the linearized operator $L_{\omega}$ can be block-diagonalized. Let $S$ be an orthogonal matrix given by
$$
S = \frac{1}{\sqrt{2}} \left (
\begin{array}{cccc} 1 & 0 & 1 & 0 \\ 0 & 1 & 0 & -1 \\ 0 & 1 & 0 & 1 \\
1 & 0 & -1 & 0 \end{array}\right).
$$
Direct computations show that
\begin{eqnarray}
\label{block1}
S^{-1} H_{\omega} S & = & \left[ \begin{array}{cc} H_+ & 0 \\
0 & H_- \end{array}\right], \\
\label{block2} S^{-1} \sigma H_{\omega} S & = & \left[ \begin{array}{cc} I & 0 \\ 0 & -I \end{array} \right] \left[
\begin{array}{cc} 0 & H_- \\ H_+ & 0 \end{array}\right],
\end{eqnarray}
where $H_{\pm}$ are two-by-two Dirac operators given by
\begin{eqnarray}
\label{dirac1}
H_{\pm} & = & \left[ \begin{array}{cc} - i \partial_x + \beta(x) - \omega & \pm (\gamma(x) - 1) \\
\pm (\gamma(x) - 1) & i \partial_x + \beta(x) - \omega
\end{array}\right] + V_{\pm}(x),
\end{eqnarray}
and $V_{\pm}(x)$ are $2$-by-$2$ matrices with exponentially decaying
coefficients given by
\begin{equation}
\label{dirac2}
V_{\pm} = \left[
\begin{array}{cc} \partial_{\bar{U} U}^2 W_N \pm \partial_{\bar{U}
\bar{V}}^2 W_N & \partial_{\bar{U} \bar{U}}^2 W_N \pm \partial_{\bar{U} V}^2 W_N \\
\partial_{U U}^2 W_N \pm \partial_{U \bar{V}}^2 W_N &
\partial_{\bar{U} U}^2 W_N \pm \partial_{U V}^2 W_N
\end{array}\right].
\end{equation}

Thanks to the symmetry of the nonlinear Dirac equations (\ref{cme}) with respect to the gauge
transformation, the linearized operator $L_{\omega}$ has a nontrivial kernel because
\begin{equation}
\label{kernel-F}
{\bf F} \equiv i \left[ \begin{array}{c} {\bf U} \\ -\bar{\bf U} \end{array} \right]
\in {\rm Ker}(L_{\omega}) \equiv
{\rm Ker}(H_{\omega}),
\end{equation}
or explicitly
\begin{equation}
\label{explicit-1}
(\ch - \omega I) {\bf U} + V_{11} {\bf U} - V_{12} \bar{\bf U} = {\bf 0}.
\end{equation}

The eigenvector ${\bf F}$ generates a two-dimensional generalized kernel
\begin{equation}
\label{kernel-F-G}
{\rm span}\{{\bf F}, {\bf G}\} \subset N_g(L_{\omega}), \quad {\bf G} = -\partial_{\omega}
\left[ \begin{array}{c} {\bf U} \\ \bar{\bf U} \end{array} \right],
\end{equation}
such that $L_{\omega} {\bf G} = {\bf F}$, or explicitly
\begin{equation}
\label{explicit-2}
(\ch - \omega I) \partial_{\omega} {\bf U} + V_{11} \partial_{\omega} {\bf U}
+ V_{12} \partial_{\omega} \bar{\bf U} = {\bf U}.
\end{equation}

The Jordan block is two-dimensional, that is, no
${\bf H} \in H^1(\R)$  solving $L_{\omega} {\bf H} = {\bf G}$ exists, if
\begin{equation}
\label{nondegeneracy}
\frac{d}{d \omega} \| {\bf U} \|^2_{L^2} \neq 0.
\end{equation}
Constraint (\ref{nondegeneracy})
is satisfied for small $a$ in Proposition \ref{proposition-local-bifurcation} under
condition (\ref{condition-one}). In the same limit,
the spectra of the linearized operator $L_{\omega}$ and the self-adjoint operator $H_{\omega}$
are characterized in the following proposition.

\begin{proposition}
Let assumptions of Proposition \ref{proposition-local-bifurcation} be satisfied.
For sufficiently small $\epsilon > 0$ and for any
$\omega \in (\omega_0, \omega_0 + \epsilon)$, we have
$$
\sigma(i L_{\omega}) = (-\infty,-1-\omega] \cup (-\infty,-1+\omega] \cup \{ 0 \}
\cup [1-\omega,\infty) \cup [1+\omega,\infty)
$$
and
$$
\sigma(H_+) = (-\infty,-1-\omega] \cup \{ \omega_1 \} \cup [1-\omega,\infty), \quad
\sigma(H_-) = (-\infty,-1-\omega] \cup \{ 0 \} \cup [1-\omega,\infty),
$$
where $\omega_1 = {\co}(|\omega - \omega_0|)$. The zero eigenvalue is
double for $L_{\omega}$ and simple for $H_-$, whereas the eigenvalue $\omega_1$ is simple.
No resonances exist at the end points of the
continuous spectrum of $L_{\omega}$ and $H_{\pm}$. \label{proposition-spectrum}
\end{proposition}

\begin{proof}
The proof holds by perturbation theory. The self-adjoint operator $H_+$ is represented by
$$
H_+ = \ch - \omega I + V_+,
$$
where $\| V_+ \|_{L^{\infty}} = {\co}(a^{2p})$ and $|\omega - \omega_0| = {\co}(a^{2p})$ as $a \to 0$ (parameter
$a$ is used in Proposition \ref{proposition-local-bifurcation}). By Assumption \ref{assumption-one},
no resonances exist in $\ch$ and, therefore, no new eigenvalues bifurcate to the gap $(-1-\omega,1-\omega)$
of the continuous spectrum of $H_+$ from the non-resonant points $\pm 1 - \omega$ for small $a > 0$.
By the perturbation theory, the only eigenvalue at $0$ for $a = 0$ becomes
the eigenvalue $\omega_1 = {\co}(a^{2p})$.

The self-adjoint operator $H_-$ is given by
$$
H_- = \sigma_3 \ch \sigma_3 - \omega I + V_-,
$$
where $\| V_- \|_{L^{\infty}} = {\co}(a^{2p})$ as $a \to 0$.
The same perturbation theory applies to self-adjoint operator $H_-$, except of the
fact that the only eigenvalue at $0$ for $a = 0$ is preserved at $0$
for $a > 0$ thanks to the gauge invariance, which results in the exact relation
$$
H_- \left[ \begin{array}{c} U \\ -V \end{array} \right] = {\bf 0}.
$$

Similarly, the double zero eigenvalue of $L_{\omega}$ is preserved at $0$
by the gauge invariance as the generalized kernel
(\ref{kernel-F-G}), whereas the continuous
spectrum does not lead to resonances at the end points or to new eigenvalues for small $a > 0$.
\end{proof}

\section{Projections and modulation equations}

By Proposition \ref{proposition-spectrum}, we have
\begin{equation}
\label{generalized-kernel}
N_g(L_{\omega}) = {\rm span}\{{\bf F}, {\bf G}\}.
\end{equation}
Recalling matrix $\sigma$ from (\ref{matrix-sigma}), we obtain the
adjoint operator
$$
L_{\omega}^* = (-i \sigma H_{\omega})^* = i H_{\omega}^* \sigma^*
= i H_{\omega} \sigma,
$$
which has the generalized kernel
\begin{equation}
\label{generalized-kernel-adjoint}
N_g(L_{\omega}^*) = {\rm span}\{\sigma {\bf F}, \sigma {\bf G} \}.
\end{equation}

Any vector $[{\bf U}_1,{\bf U}_2]^T$
in the invariant subspace of the linearized operator $L_{\omega}$ in $L^2(\R)$,
which is an orthogonal complement of the generalized null space $N_g(L_{\omega})$,
has to satisfy the symplectic orthogonality conditions
\begin{equation}
\label{symplectic-orthogonality}
\left\{ \begin{array}{l} \langle {\bf U}, {\bf U}_1 \rangle_{L^2} +
\langle \bar{\bf U},{\bf U}_2 \rangle_{L^2} = 0, \\
\langle \partial_{\omega} {\bf U}, {\bf U}_1 \rangle_{L^2} -
\langle \partial_{\omega} \bar{\bf U},{\bf U}_2 \rangle_{L^2}= 0. \end{array} \right.
\end{equation}
If ${\bf U}_2 = \bar{\bf U}_1$, the symplectic orthogonality conditions (\ref{symplectic-orthogonality})
can be rewritten in the explicit form
\begin{equation}
\label{symplectic-orthogonality-2}
{\rm Re} \langle {\bf U}, {\bf U}_1 \rangle_{L^2} = 0, \quad
{\rm Im} \langle \partial_{\omega} {\bf U}, {\bf U}_1 \rangle_{L^2} = 0.
\end{equation}

Using symplectic orthogonality conditions, we now set up modulation
equations for nonlinear dynamics of small gap solitons. By
Proposition \ref{proposition-local-existence}, we have at least local solutions
of the nonlinear Dirac equations (\ref{cme-formal}). Now we look for local
solutions in the form
\begin{equation}
\label{modulation}
\left\{ \begin{array}{cc}
u(x,t) = e^{-i \theta(t)} \left[ U(x;\omega(t)) + U_1(x,t) \right], \\
v(x,t) = e^{-i \theta(t)} \left[ V(x;\omega(t)) + V_1(x,t) \right],
\end{array} \right.
\end{equation}
where we write explicitly the dependence of the stationary solution ${\bf U} = [U,V]^T$ on $\omega$.
The time evolution problem for ${\bf U}_1 = [U_1,V_1]^T$ is given by
\begin{equation}
\label{time-evolution-vector}
i \frac{d {\bf U}_1}{d t} = (\ch -\omega I) {\bf U}_1 - i \dot{\omega} \partial_{\omega} {\bf U}
- (\dot{\theta}-\omega) ({\bf U} + {\bf U}_1) + {\bf N}({\bf U}+{\bf U}_1) - {\bf N}({\bf U}).
\end{equation}

Using the symplectic orthogonality condition (\ref{symplectic-orthogonality-2}) on ${\bf U}$, we obtain the
modulation equations on $\omega(t)$ and $\theta(t)$:
\begin{eqnarray}
\label{projection-equations}
\left\{ \begin{array}{cc} \dot{\omega} {\rm Re} \langle \partial_{\omega} {\bf U}, {\bf U} - {\bf U}_1 \rangle_{L^2} +
(\dot{\theta} - \omega) {\rm Im} \langle {\bf U}, {\bf U}_1 \rangle_{L^2} = F_1, \\
\dot{\omega} {\rm Im} \langle \partial^2_{\omega} {\bf U}, {\bf U}_1 \rangle_{L^2} +
(\dot{\theta} - \omega) {\rm Re} \langle \partial_{\omega} {\bf U}, {\bf U} + {\bf U}_1 \rangle_{L^2} = F_2,
\end{array} \right.
\end{eqnarray}
where
\begin{eqnarray*}
F_1 & = & {\rm Im} \left[ \langle {\bf U}, {\bf N}({\bf U}+{\bf U}_1) - {\bf N}({\bf U}) \rangle_{L^2}
+ \langle \bar{V}_{12} \bar{\bf U} - V_{11} {\bf U}, {\bf U}_1 \rangle_{L^2} \right], \\
F_2 & = & {\rm Re} \left[ \langle \partial_{\omega} {\bf U}, {\bf N}({\bf U}+{\bf U}_1) - {\bf N}({\bf U}) \rangle_{L^2} - \langle V_{12} \partial_{\omega} \bar{\bf U} + V_{11}  \partial_{\omega} {\bf U}, {\bf U}_1 \rangle_{L^2} \right]
\end{eqnarray*}
and equations (\ref{explicit-1}) and (\ref{explicit-2}) have been
used. The following result shows that the right-hand side of
system (\ref{projection-equations}) is quadratic with respect to
the perturbation vector ${\bf U}_1$.

\begin{proposition}
Let assumptions of Proposition \ref{proposition-local-bifurcation} be satisfied.
Fix small $\epsilon > 0$ and $\delta > 0$
For any $\omega \in (\omega_0, \omega_0 + \epsilon)$ and
any ${\bf U}_1 \in B_{\delta}(L^{\infty})$, there is $C_{\epsilon,\delta} > 0$ such that
\begin{equation}
\label{e:1} |F_1| + |F_2| \leq C_{\epsilon,\delta} \langle {\bf U}_1^2, {\bf U} \rangle_{L^2}.
\end{equation}
\label{proposition-vector-field}
\end{proposition}

\begin{proof}
We use the expansion (\ref{expansion-N}), constraints on matrices (\ref{potentials-explicit-linearized})
and elementary properties of inner product (\ref{norm-l2}) to show that the linear terms
in ${\bf U}_1$ vanish in the expression for $F_1$ and $F_2$. For instance, $F_1$
has the following linear terms in ${\bf U}_1$:
\begin{eqnarray*}
{\rm Im} \left[ \langle {\bf U}, V_{11} {\bf U}_1 + V_{12} \bar{\bf U}_1 \rangle_{L^2}
+ \langle \bar{V}_{12} \bar{\bf U} - V_{11} {\bf U}, {\bf U}_1 \rangle_{L^2} \right] \\
=
{\rm Im} \left[ \langle \bar{V}_{11}^T {\bf U},  {\bf U}_1 \rangle_{L^2}
- \langle V_{11} {\bf U}, {\bf U}_1 \rangle_{L^2}
+ \langle \bar{V}_{12}^T {\bf U}, \bar{\bf U}_1 \rangle_{L^2}
+ \langle \bar{V}_{12} \bar{\bf U},
{\bf U}_1  \rangle_{L^2}   \right] = 0.
\end{eqnarray*}
Similar computations holds for linear terms of $F_2$. Together
with smoothness in Assumption \ref{assumption-two}, this computation
shows that both terms $F_1, F_2$ are quadratic in ${\bf U_1}$ in the
sense of \eqref{e:1}.
\end{proof}

Setting ${\bf U}_1 = {\bf Y} e^{i \theta}$, we rewrite the
time-evolution equation (\ref{time-evolution-vector}) in the
equivalent form
\begin{equation}
\label{residual-equation} i \frac{d {\bf Y}}{d t} = \ch {\bf Y} +
e^{-i \theta} {\bf F}, \quad {\bf F} = - i \dot{\omega}
\partial_{\omega} {\bf U} - (\dot{\theta}-\omega) {\bf U} + {\bf
N}({\bf U}+{\bf Y} e^{i \theta}) - {\bf N}({\bf U}).
\end{equation}

We are now ready to formulate the main theorem of this article.

\begin{theorem}
\label{theorem-main} Assume Assumption \ref{assumption-one}, Assumption \ref{assumption-two}
with $p \geq 2$, and condition (\ref{condition-one}). Fix $\epsilon > 0$ and $\delta > 0$ sufficiently small such that
$\theta(0) = 0$, $\omega(0) \in (\omega_0,\omega_0 + \epsilon)$, and
${\bf Y}(0) \in B_{\delta}(H^1)$. There exist $\epsilon_0 > \epsilon$, $\theta_{\infty} \in \mathbb{R}$,
$\omega_{\infty} \in (\omega_0,\omega_0 + \epsilon_0)$,
$(\omega,\theta) \in C^1(\mathbb{R}_+,\mathbb{R}^2)$, and
$$
{\bf Y}(t) \in C(\mathbb{R}_+,H^1) \cap L^4(\mathbb{R}_+,L^\infty)
$$
such that $(\omega,\theta)(t)$ solve the modulation equations (\ref{projection-equations}),
${\bf Y}(t)$ solves the evolution equation (\ref{residual-equation}), and
$$
\lim_{t \to \infty} \left( \theta(t) - \int_0^t \omega(s) ds
\right) = \theta_{\infty}, \quad \lim_{t \to \infty} \omega(t) =
\omega_{\infty}, \quad \lim_{t \to \infty} \| {\bf Y}(t)
\|_{L^{\infty}} = 0.
$$
\end{theorem}

We shall prove this theorem in the remainder of the article. To do so, we shall develop first the
spectral theory for the Dirac operator $\ch$ and obtain the dispersive decay estimates
for the semi-group $e^{- i t \ch}$ acting on the continuous spectrum of $\ch$.

\section{Spectral theory for operator $\ch$}
\label{sec:3}

Let us consider the spectral problem $\ch {\bf u} = \lambda {\bf u}$ or explicitly,
\begin{eqnarray}
\left\{ \begin{array}{l} -i u'(x) + \beta(x) u(x) + (\gamma(x) - 1) v(x) = \lambda u(x), \\
i v'(x) + \beta(x) v(x) + (\gamma(x) - 1) u(x) = \lambda v(x), \end{array} \right. \quad x \in \R.
\label{spectral-Dirac}
\end{eqnarray}
Recall that
$$
\sigma_c(\ch) \equiv (-\infty,-1] \cup [1,\infty).
$$

Here we develop the scattering theory of wave operators for the
Dirac operator $\ch$. A similar theory for the Schr\"{o}dinger
operators on an infinite line goes back to the works of Weder
\cite{Weder, Weder2} and Goldberg \& Schlag \cite{GS}.

Let us first define the Jost functions for $\lambda \in (-\infty,-1]$
at one branch of $\sigma_c(\ch)$. To do so, let us
parameterize $(-\infty,-1]$ by $\lambda = -\sqrt{1 + k^2}$ for
$k \in \R$ and consider solutions of system (\ref{spectral-Dirac}) according to
the boundary conditions
\begin{eqnarray}
\label{boundary-values-u} {\bf u}^{\pm}(x;k) \to \left[
\begin{array}{c} 1 \\ \alpha_{\pm}(k) \end{array} \right]
e^{\pm i k x} \quad \mbox{\rm as} \quad x \to \pm \infty,
\end{eqnarray}
where $\alpha_{\pm}(k) := \sqrt{1 + k^2} \pm k$. The following
proposition gives the construction of Jost functions.

\begin{proposition}
For any $k \in \R$, there exist unique Jost functions ${\bf
u}^{\pm}(x;k)$ such that
$$
\lim_{x \to \pm \infty} \left[ {\bf u}^{\pm}(\cdot;k) -
[1,\alpha_+]^T e^{\pm i k x} \right] = 0.
$$
Moreover,

\begin{itemize}
\item If $k \neq 0$, then ${\bf u}^{\pm}(\cdot;k) \in L^{\infty}(\R)$.

\item If $k = 0$, then ${\bf u}^{\pm}(x;0)$ may grow  at most linearly in $x$ as $x \to \mp \infty$.

\item As $k \to \pm \infty$, both ${\bf u}^+(x;k)_1$ and ${\bf
u}^-(x;k)_1$ are bounded, ${\bf u}^{\pm}(x;k)_2$ grows linearly
in $k$, and ${\bf u}^{\mp}(x;k)_2$ decays inverse linearly in
$k$.
\end{itemize}
\label{proposition-Jost-functions}
\end{proposition}

\begin{proof}
Setting ${\bf u}^{\pm}(x;k) = {\bf m}^{\pm}(x;k) e^{\pm i k x}$
and using the Green function technique, we obtain an integral
equation for the Jost functions ${\bf m}^{\pm}(x;k)$
\begin{equation}
\label{integral-equations-Volterra}
{\bf m}^{\pm}(x;k) =  \left[ \begin{array}{c} 1 \\ \alpha_{\pm}
\end{array} \right] + \int_x^{\pm \infty} G^{\pm}(x-y;k) V(y) {\bf
m}^{\pm}(y;k)dy,
\end{equation}
where
\begin{equation}
\label{matrix-potential}
G^{\pm}(x;k) =
\frac{1}{2 i k} \left[ \begin{array}{cc} \alpha_{\mp} - \alpha_{\pm}
e^{\mp 2 i k x} & 1 - e^{\mp 2i k x} \\
1 - e^{\mp 2i k x}  & \alpha_{\pm} - \alpha_{\mp} e^{\mp 2 i k x}
\end{array} \right].
\end{equation}

Under the assumption of fast decay of $V(x)$ to $0$
as $|x| \to \infty$, the standard theory gives solutions ${\bf m}^{\pm}(\cdot;k) \in
L^{\infty}(\R)$ of the integral equations (\ref{integral-equations-Volterra}) for
$k \neq 0$ and the scattering relation between the solutions
for all $x \in \R$ including $x \to \mp \infty$
\begin{equation}
\left\{ \begin{array}{l} {\bf m}^+(x;k) = a^+(k) {\bf m}^-(x;-k) +
b^+(k) {\bf m}^-(x;k) e^{-2 i k x}, \\
{\bf m}^-(x;k) = a^-(k) {\bf m}^+(x;-k) + b^-(k) {\bf
m}^+(x;k) e^{2 i k x},  \end{array} \right. \label{scattering-relation}
\end{equation}
where
\begin{eqnarray*}
a^{\pm}(k) & = & 1 \pm \frac{1}{2 i k} \int_{\R} \left( \alpha_{\mp} [V(x)
{\bf m}^{\pm}(x;k)]_1 +  [V(x) {\bf m}^{\pm}(x;k)]_2 \right) dx, \\
b^{\pm}(k) & = & \mp \frac{1}{2 i k} \int_{\R} \left( \alpha_{\pm} [V(x) {\bf
m}^{\pm}(x;k)]_1 +  [V(x) {\bf m}^{\pm}(x;k)]_2 \right) e^{\pm 2 i k x} dx.
\end{eqnarray*}
It follows from the explicit expressions for $a^{\pm}(k)$ and $b^{\pm}(k)$ that
\begin{eqnarray}
a^{\pm}(k) \to \frac{\pm \gamma^{\pm}}{2 i k}, \quad b^{\pm}(k)
\to \frac{\mp \gamma^{\pm}}{2 i k} \quad \mbox{\rm as} \quad k \to
0, \label{limit-a-b}
\end{eqnarray}
where
\begin{eqnarray}
\label{gamma-resonance}
\gamma^{\pm} = \int_{\R} (\beta(x) + \gamma(x)) (m_1^{\pm}(x;0) + m_2^{\pm}(x;0)) dx.
\end{eqnarray}

If $k = 0$, the Jost functions ${\bf m}^{\pm}(x;0)$ satisfy the integral equation
\begin{equation}
\label{integral-equations-Volterra-2}
{\bf m}^{\pm}(x;0) =  \left[ \begin{array}{c} 1 \\ 1
\end{array} \right] + \int_x^{\pm \infty} G^{\pm}(x-y;0) V(y) {\bf
m}^{\pm}(y;0)dy,
\end{equation}
where
\begin{equation}
G^{\pm}(x;0) = \pm \left[ \begin{array}{cc} x + i & x \\
x  & x - i \end{array}
\right].
\end{equation}
Thanks to the fast decay of $V(x)$, existence of locally bounded
function ${\bf m}^{\pm}(x;0)$ follows again from the standard theory. The linear growth of
${\bf m}^{\pm}(x;0)$ as $x \to \mp \infty$ follows from the integral equations (\ref{integral-equations-Volterra-2})
if $\gamma^{\pm} \neq 0$.

Finally, as $k \to \pm \infty$, $\alpha_{\pm}$ grows linearly in
$k$, $\alpha_{\mp}$ decays inverse linearly in $k$, whereas
$G^{\pm}(x;k)$ remains bounded. The asymptotic behavior of
solutions ${\bf m}^{\pm}(x;k)$ of the integral equations
(\ref{integral-equations-Volterra}) follows the asymptotic
behavior of the limiting functions (\ref{boundary-values-u}) in
$k$ as $k \to +\infty$.
\end{proof}

\begin{remark}
Proposition \ref{proposition-Jost-functions} eliminates the possibility of embedded eigenvalues
in the continuous spectrum $\sigma_c(\ch)$ because the space of solutions of the
Dirac system (\ref{spectral-Dirac}) for $\lambda < -1$ is spanned by
the two fundamental solutions ${\bf u}^{\pm}(x;k)$ with no decay to zero as $x \to \pm \infty$.
\end{remark}

The following proposition summarizes the relations on the scattering coefficients in the scattering relation (\ref{scattering-relation}).

\begin{proposition}
For any $k \in \R$, we have
\begin{equation}
\label{scat-rel-1}
a^+(k) = a^-(k), \quad b^+(k) = -b^-(-k),
\end{equation}
\begin{equation}
\label{scat-rel-2}
a^+(-k) = \bar{a}^+(k), \quad b^+(-k) = \frac{\sqrt{1 + k^2} - k}{\sqrt{1 + k^2} + k} \bar{b}^+(k),
\end{equation}
and
\begin{equation}
\label{scat-rel-3}
|a^+(k)|^2 = 1 + \frac{\sqrt{1 + k^2} - k}{\sqrt{1 + k^2} + k} |b^+(k)|^2.
\end{equation}
\end{proposition}

\begin{proof}
Inverting the scattering relation (\ref{scattering-relation}), we obtain the constraint
on the scattering coefficients for all $k \in \R$
\begin{eqnarray}
\left\{ \begin{array}{l} a^+(k) a^-(-k) + b^+(k) b^-(k) = 1, \\
a^+(k) b^-(-k) + b^+(k) a^-(k) = 0. \end{array} \right.
\label{relation-scattering-data}
\end{eqnarray}

Let $W({\bf u}_1,{\bf u}_2)$ denote the Wronskian determinant
of any two solutions ${\bf u}_1$ and ${\bf u}_2$ of the Dirac system (\ref{spectral-Dirac}).
It is clear that $W({\bf u}_1,{\bf u}_2)$ is constant in $x \in \R$. Therefore,
$W({\bf u}_1,{\bf u}_2)$ can be computed in the limits $x \to \pm \infty$. Using boundary values (\ref{boundary-values-u}) and scattering relations (\ref{scattering-relation}), we obtain
\begin{eqnarray}
\label{Wronskian}
W({\bf u}^+,{\bf u}^-) = u^+_1(x;k) u^-_2(x;k) - u^+_2(x;k) u^-_1(x;k) = -2k a^+(k) = -2k a^-(k).
\end{eqnarray}
This result together with the second equation of system (\ref{relation-scattering-data}) gives relations (\ref{scat-rel-1}). The first equation of system (\ref{relation-scattering-data}) implies
now for all $k \in \R$ that
\begin{equation}
\label{scat-rel-4}
a^+(k) a^+(-k) - b^+(k) b^+(-k) = 1.
\end{equation}

Let $(u_k,v_k)$ denote any solution of the Dirac system (\ref{spectral-Dirac}) for
$\lambda = -\sqrt{1 + k^2}$. It is checked directly that
$$
|u_k|^2 - |v_k|^2 \quad \mbox{\rm and} \quad \bar{u}_{-k} u_k - \bar{v}_{-k} v_k
$$
are constant in $x$. Using boundary values (\ref{boundary-values-u}) and
scattering relations (\ref{scattering-relation}) again, we obtain
\begin{eqnarray}
\label{relation-1}
(\sqrt{1 + k^2} +k) (1 - |a^+(k)|^2) + (\sqrt{1 + k^2} -k) |b^+(k)|^2 & = & 0,\\
\label{relation-2}
(\sqrt{1 + k^2} +k) \overline{b^-(k)} + (\sqrt{1 + k^2} -k) b^+(k) & = & 0.
\end{eqnarray}
These identities together with equation (\ref{scat-rel-4}) give relations (\ref{scat-rel-2})
and (\ref{scat-rel-3}).
\end{proof}

\begin{remark}
Identity (\ref{scat-rel-3}) shows that $|a^+(k)| \geq 1$ for all $k \in \R$. This excludes embedded
resonant states with $a^+(k) = 0$. There is still a possibility of end-point resonances at
$k = 0$, since $a^+(k)$ is generally singular as $k \to 0$. We say that the end points $\pm 1$ are resonances
if there exist a solution ${\bf u} \in L^{\infty}(\R)$ of the spectral problem (\ref{spectral-Dirac})
for $\lambda = \pm 1$. If this is the case, then $\gamma^+ = 0$ and $\lim_{k \to 0} a^+(k)$ exists.
\label{remark-resonance}
\end{remark}

We shall now define the Jost functions for $\lambda \in [1,\infty)$ at the other branch
of $\sigma_c(\ch)$. Similarly to the analysis for $\lambda \in (-\infty,-1]$,
we can parameterize $[1,\infty)$ by $\lambda = \sqrt{1 + k^2}$ for $k \in \R$
and consider solutions of system (\ref{spectral-Dirac}) according to
the boundary conditions
\begin{eqnarray*}
{\bf v}^{\pm}(x;k) \to \left[ \begin{array}{c} -\alpha_{\pm} \\ 1
\end{array} \right] e^{\pm i k x} \quad \mbox{\rm
as} \quad x \to \pm \infty.
\end{eqnarray*}
Using a similar Green's function formulation, Proposition \ref{proposition-Jost-functions}
can be extended to functions ${\bf v}^{\pm}(x;k)$. In what follows, we will not treat
functions ${\bf v}^{\pm}(x;k)$ for $\lambda \in [1,\infty)$ but will only be working with
functions ${\bf u}^{\pm}(x;k)$ for $\lambda \in (-\infty,-1]$. This approach does not limit
any generality. Moreover, we note the particularly remarkable case.

\begin{remark}
If $\beta(x) \equiv 0$, the Jost functions are related by
$$
{\bf v}^{\pm}(x;k) = \left[ \begin{array}{cc} 0 & -1 \\ 1 & 0 \end{array} \right] {\bf u}^{\pm}(x;k),
$$
thanks to the symmetry of the Dirac system (\ref{spectral-Dirac}).
\end{remark}

Let $R_{\ch}(\lambda) = (\ch - \lambda I)^{-1}$ be the resolvent operator, defined as a bounded
operator from $L^2(\R)$ to $L^2(\R)$ for any $\lambda \notin \sigma(\ch)$. Using
the Jost functions and the standard limiting absorption principle,
the resolvent operator is extended to the continuous spectrum
as a uniformly bounded operator from $L^2_{\alpha}(\R)$ to $L^2_{-\alpha}(\R)$ for any $\alpha > \frac{1}{2}$.
Let us denote the limiting operators by
$$
R_{\ch}^{\pm}(\la) := \lim_{\epsilon \downarrow 0} R_{\ch}(\la \pm i \epsilon), \quad \lambda \in \sigma_c(\ch),
$$
depending on whether $\lambda \to \sigma_c(\ch)$ from the upper or lower halves of the complex plane of $\lambda$.

The following proposition allows us to express
$R_{\ch}^{\pm}(\lambda)$  for $\lambda \in \sigma_c(\ch)$ in terms
of the Jost functions. According to the previous remarks, it is
sufficient to consider $\lambda \in (-\infty,-1]$. The arguments
for $\lambda \in [1,\infty)$ can be developed similarly.

\begin{proposition}
\label{proposition-resolvent} For any $\lambda \in (-\infty,-1]$
and any fixed $\alpha
> \frac{1}{2}$, operators $R_{\ch}^{\pm}(\lambda) :
L^2_{\alpha}(\R) \mapsto L^2_{-\alpha}(\R)$ can be represented by
the integral kernel in the form
\begin{equation}
[R^{\pm}_{\ch}(\lambda)](x,y) = \frac{\pm 1}{2 i k a^+(\pm k)} \left\{ \begin{array}{l}
{\bf u}^+(x;\pm k) [\sigma_1 {\bf u}^-(y;\pm k)]^T, \quad x > y, \\
{\bf u}^-(x;\pm k) [\sigma_1 {\bf u}^+(y;\pm k)]^T, \quad x < y,
\end{array} \right.  \label{resolvent-limiting}
\end{equation}
where $k \leq 0$ and $\lambda = -\sqrt{1 + k^2}$.
\end{proposition}

\begin{proof}
Let us consider the solutions of the linear system for a fixed $y \in \R$ and $\lambda = -\sqrt{1 + k^2}$,
\begin{eqnarray}
(\ch - \lambda I) [R^+_{\ch}(\lambda)](x,y) = \delta(x-y) Id,
\label{spectral-Dirac-delta}
\end{eqnarray}
which satisfy the asymptotic behavior,
\begin{equation}
\label{asymptotic-behavior} [R^+_{\ch}(\lambda)](x,y) \sim e^{i k
|x-y|}, \quad \mbox{\rm as} \quad |x-y| \to \infty.
\end{equation}
The function $[R^+_{\ch}(\lambda)](x,y)$ decays exponentially
as $|x-y| \to \infty$ if $k$ is extended off the
real axis with ${\rm Im}(k) > 0$. Since ${\rm Re}(\lambda) {\rm
Im}(\lambda) = {\rm Re}(k) {\rm Im}(k)$ and ${\rm Re}(\lambda)
\leq -1$, we understand that the behavior (\ref{asymptotic-behavior}) recovers
the limiting resolvent operator $R^+_{\ch}(\lambda)$ defined for ${\rm Im}(\lambda) \geq 0$ if ${\rm
Re}(k) \leq 0$.

For the first column vector of the linear system (\ref{spectral-Dirac-delta}), denoted by $(u,v)$, we obtain
\begin{equation}
\label{column-representation} \left[ \begin{array}{c} u \\ v
\end{array} \right] = \left\{ \begin{array}{l}
c(y,k) {\bf u}^+(x;k), \quad x > y, \\
d(y,k) {\bf u}^-(x;k), \quad x < y, \end{array} \right.
\end{equation}
where the behavior (\ref{asymptotic-behavior}) is
satisfied thanks to the boundary conditions
(\ref{boundary-values-u}). Parameters $(c,d)$ are to be determined.

Matching conditions across the point $x = y$ sets up the linear system for $c$ and $d$ with the unique solution,
$$
c(y,k) = \frac{i {\bf u}^-(y;k)_2}{W({\bf u}^+,{\bf u}^-)}, \quad
d(y,k) = \frac{i {\bf u}^+(y;k)_2}{W({\bf u}^+,{\bf u}^-)},
$$
where $W({\bf u}^+,{\bf u}^-) = -2 k a^+(k)$ by identity (\ref{Wronskian}).

Similarly for the second column vector of the linear system
(\ref{spectral-Dirac-delta}), we obtain the same expression
(\ref{column-representation}) with a different solution of the
linear system for $(c,d)$,
$$
c(y,k) = \frac{i {\bf u}^-(y;k)_1}{W({\bf u}^+,{\bf u}^-)}, \quad
d(y,k) = \frac{i {\bf u}^+(y;k)_1}{W({\bf u}^+,{\bf u}^-)}.
$$
Using the Pauli matrix $\sigma_1$, we arrive to the expression
(\ref{resolvent-limiting}) for $R^+_{\ch}(\lambda)$. The
expression for $R^-_{\ch}(\lambda)$ is found by the replacement of
$k$ by $-k$. The exponential decay as $|x-y| \to \infty$ occurs now
for ${\rm Im}(k) < 0$. The limiting resolvent operator $R^-_{\ch}(\lambda)$ is defined for
${\rm Im}(\lambda) \leq 0$ if ${\rm Re}(k) \leq 0$.
\end{proof}

The following proposition describes $\lambda$-uniform bounds on
the limiting resolvent operators $R_{\ch}^{\pm}(\lambda)$ in
weighted spaces. In order to exclude problems at the end points
$\lambda = \pm 1$, we assume that no end-point resonances occur
at $k = 0$ (Assumption \ref{assumption-one}). Thanks to Remark \ref{remark-resonance},
it is equivalent to assume that $\gamma^+ \neq 0$.

\begin{proposition}
\label{prop:a1} Let $\gamma^+ \neq 0$ in (\ref{gamma-resonance}). For any $\alpha > \frac{3}{2}$,
there exists constant $C_{\alpha} > 0$ such that
\begin{eqnarray}
\label{b:30} \sup_{|\la| \geq 1}
\|R_{\ch}^{\pm}(\la)\|_{L^2_{\alpha} \to L^2_{-\al}} \leq
C_{\alpha}.
\end{eqnarray}
In addition, for any $\alpha \geq 1$, there exists constant
$C_{\alpha} > 0$ such that
\begin{eqnarray}
\label{b:20} \sup_{|\la| \geq 1} \|  R_{\ch}^{\pm}(\la)
\|_{L^1_{\alpha}  \to L^\infty_{-\alpha}} \leq C_{\alpha}.
\end{eqnarray}
\end{proposition}

\begin{proof}
We recall that $|a^+(k)| \geq 1$ (Remark \ref{remark-resonance}).
Thanks to the asymptotic expansion (\ref{limit-a-b}), if $\gamma^+
\neq 0$, then $k a^+(k) \neq 0$ for any $k \in \R$. Using this
result and Proposition \ref{proposition-resolvent}, we construct
$$
\hat{R}^{\pm}_{\ch,\al}(x,y) \equiv
\frac{[R_{\ch}^{\pm}(\lambda)](x,y)}{(1+x^2)^{\alpha/2}
(1+y^2)^{\alpha/2}}.
$$

By Proposition \ref{proposition-Jost-functions}, ${\bf
u}^{\pm}(\cdot;k) \in L^{\infty}(\R)$ for every $k \neq 0$ and
${\bf u}^{\pm}(x;0)$ grow at most linearly in $x$ as $x \to \mp
\infty$. Therefore, $\hat{R}^{\pm}_{\ch,\al}(x,y)$ is a kernel of
a Hilbert--Schmidt operator for any fixed $\lambda \in
(-\infty,-1]$ and $\alpha > \frac{3}{2}$.

It remains to show that $\hat{R}^{\pm}_{\ch,\al}(x,y)$ is
uniformly bounded in the limit $\lambda \to -\infty$ ($k \to
-\infty$) for any $x,y \in \R$. Note that
$$
{\bf u}^+(x;k) [\sigma_1 {\bf u}^-(y;k)]^T = \left[ \begin{array}{ccc} {\bf u}^+(x;k)_1 {\bf u}^-(y;k)_2 &
{\bf u}^+(x;k)_1 {\bf u}^-(y;k)_1 \\
{\bf u}^+(x;k)_2 {\bf u}^-(y;k)_2 & {\bf u}^+(x;k)_2 {\bf u}^-(y;k)_1
 \end{array} \right]
$$
and a similar formula for ${\bf u}^-(x;k) [\sigma_1 {\bf
u}^+(y;k)]^T$. By Proposition \ref{proposition-Jost-functions},
this matrix grows linearly in $k$ as $k \to -\infty$ for any $x, y
\in \R$. On the other hand, $k a_+(k)$ grows at least linearly as
$|k| \to \infty$, which implies the $\lambda$-uniform bound \eqref{b:30}.

To prove bound \eqref{b:20}, we can see from the linear growth of
${\bf u}^{\pm}(x;0)$ as $x \to \mp \infty$ that
$\hat{R}^{\pm}_{\ch,\al}(x,y)$ is a kernel of a bounded operator
from $L^1(\R)$ to $L^{\infty}(\R)$ for any $\alpha \geq 1$. The
mapping is also bounded as $k \to -\infty$.
\end{proof}

Let $P_{a.c.}(\ch) : L^2(\R) \mapsto L^2(\R)$ be the orthogonal
projection operator to the continuous spectrum of $\ch$. We recall
the Cauchy formula,
\begin{equation}
\label{Cauchy-formula} e^{-i t \ch} P_{a.c.}(\ch) f = \f{1}{2\pi
i} \left( \int_{-\infty}^{-1} + \int_{1}^{\infty} \right) e^{-i t
\la} \left[ R^+_{\ch}(\la) - R^-_{\ch}(\la) \right] f d\la,
\end{equation}
where the integral is understood in the norm of the mapping from
$L^2_{\al}(\R)$ to $L^2_{-\al}(\R)$ for $\alpha > \frac{3}{2}$.
The interval $(-\infty,-1]$ for $\lambda$ can be parameterized by
$(-\infty,0]$ for $k$ using the substitution
$$
\lambda = -\sqrt{1 + k^2} \quad \Rightarrow \quad d \lambda =
-\frac{k dk}{\sqrt{1 + k^2}}.
$$
These representations are used for the derivation of linear
dispersive decay estimates for the semi-group $e^{-it \ch} P_{\rm
a.c.}(\ch)$.

\section{Linear estimates for the operator $\ch$}

We shall need two preliminary results, which will be useful
in our arguments for this section.

\subsection {Preliminaries}

The first result that we need  is the Christ-Kiselev lemma.
We actually state a  version due to Smith \& Sogge \cite{SS}.

\begin{lemma}
\label{christ_kiselev}
Let $X,Y$ be Banach spaces and
$\ck:L^p(\R;X)\to L^q(\R,Y)$ be a linear operator such
that $\ck f(t)=\int_{-\infty}^\infty K(t,s) f(s) ds$. Then, the operator
\begin{equation}\label{definition-tilde-K}
\tilde{\ck} f(t)=\int_0^t K(t,s) f(s) ds,
\end{equation}
is bounded from $L^p(\R;X)$ to $L^q(\R,Y)$, provided $p<q$. Moreover, there is
$C_{p,q} > 0$ such that
$$
\|\tilde{\ck}\|_{L^p(\R;X)\to L^q(\R,Y)} \leq C_{p,q}
\|\ck\|_{L^p(\R;X)\to L^q(\R,Y)}.
$$
\end{lemma}
The second lemma is a technical statement, which is complementary
to Lemma \ref{christ_kiselev}, when the condition $p<q$ is
violated (most notably when $p=q$).
This is stated for the Schr\"odinger operator $-\p_x^2 + V(x)$ by Mizumachi
(Lemma 11 in \cite{Mizum}), but it applies equally well to an arbitrary
self-adjoint operator $\cl$.

\begin{lemma}
\label{mizumachi_1} Let $\cl$ be a self-adjoint operator and
$P_{a.c}(\cl)$ be a projection to the absolute continuous spectrum
of $\cl$. Let $g(t,x)=g_1(t)g_2(x)$ and define the function
\begin{eqnarray}
\label{function-U-new}
U(t,x)=\f{i}{\sqrt{2\pi}}
\int_{-\infty}^\infty e^{-i t \la}\check{g_1}(\la) \left(\left[
R^+_{\cl}(\la) + R^-_{\cl}(\la) \right]  g_2\right)(x) d\la.
\end{eqnarray}
Then, we have
$$
U(t,\cdot) = 2\int_0^t e^{-i(t-s)\cl} P_{a.c.}(\cl) g(s, \cdot)
ds+ \left( \int_{-\infty}^0 - \int_{0}^\infty \right)
e^{-i(t-s)\cl} P_{a.c.}(\cl) g(s, \cdot)ds.
$$
\end{lemma}

We use the resolvent analysis of the Dirac operator $\ch$ to
derive some linear estimates, which are used in the proof of the
main theorem.

\subsection{Mizumachi estimates}

We refer to Mizumachi's work \cite{Mizum} in the context of the
one-dimensional NLS equation, which was used in our work
\cite{KPS} in the context of the discrete NLS equation. These
estimates are developed to control quadratic nonlinearities in the
time-evolution equation (\ref{time-evolution-vector}), which have
fast spatial decay. Thus, the challenge here is to achieve $L^2_t$
temporal decay, in the presence of the exponential spatial decay.

\begin{lemma}
\label{le:2} Fix $\al > \frac{3}{2}$. There is $C_{\al} > 0$ such
that
\begin{equation}
\label{a:10} \|\langle x\rangle^{-\al} e^{-i t \ch} P_{a.c.}(\ch)
f\|_{L^\infty_x L^2_t}\leq C_{\al} \|f\|_{L^2_x}
\end{equation}
and
\begin{equation}
\label{b:15} \left\| \langle x\rangle^{-\al} \int_0^t
e^{-i(t-\tau)\ch} P_{a.c.}(\ch)F(\tau, \cdot)d\tau
\right\|_{L^\infty_x L^2_t}\leq C_{\al} \|\langle x\rangle^{\al}
F\|_{L^1_x L^2_t}.
\end{equation}
\end{lemma}

\begin{proof}
The proof of Lemma \ref{le:2} proceeds via analysis of the contribution of the high energy part
and the low energy part.

Let $\chi(x)$ be an even $C^\infty$ function with $\chi(x)=1$ for
$|x|<1$ and $\chi(x)=0$ for $|x|>2$. Fix $M > 2$, let $\chi_M(x)=\chi(x/M)$
and decompose
$$
e^{-i t \ch} P_{a.c.}(\ch)  f = \chi_M e^{-i t \ch} P_{a.c.}(\ch)
f + (1-\chi_M)  e^{-i t \ch} P_{a.c.}(\ch) f.
$$
In order to show \eqref{a:10}, we need the following two estimates
\begin{eqnarray}
\label{a:11}
& & \| (1-\chi_M) e^{-i t \ch} P_{a.c.}(\ch) f\|_{L^\infty_x L^2_t}\leq C \|f\|_{L^2_x} \\
\label{a:12} & & \|<x>^{-\al} \chi_M  e^{-i t \ch} P_{a.c.}(\ch)
f\|_{L^\infty_x L^2_t}\leq C \|f\|_{L^2_x}
\end{eqnarray}
Combining bounds (\ref{a:11}) and (\ref{a:12}), we complete the
proof of estimate \eqref{a:10}. Bounds (\ref{a:11}) and
(\ref{a:12}) are proven in the following two subsections.

The proof of estimate \eqref{b:15} is based upon Proposition \ref{prop:a1} and
Lemma \ref{mizumachi_1}. By Lemma \ref{mizumachi_1}, we can write (with $\cl=\ch$)
\begin{eqnarray*}
\int_0^t e^{-i(t-\tau)\ch} P_{a.c.}(\ch) F(\tau,
\cdot)d\tau=\f{1}{2}U + \f{1}{2} \left( \int_0^\infty -
\int_{-\infty}^0 \right) e^{-i(t-\tau)\ch} P_{a.c.}(\ch) F(\tau,
\cdot)d\tau,
\end{eqnarray*}
where $U$ is the function defined by (\ref{function-U-new}). Let
us first control the last two terms. Since they are similar, we
only need to control one of the terms. By the estimate
\eqref{a:10}, we have
\begin{eqnarray*}
\left\| \langle x\rangle^{-\al} e^{-it \ch}\int_0^\infty  e^{i
\tau \ch} P_{a.c.}(\ch)F(\tau, \cdot) d\tau \right\|_{L^\infty_x
L^2_t} & \leq &
C \left\| \int_0^\infty  e^{i \tau \ch}  P_{a.c.}(\ch)F(\tau, \cdot)d\tau \right\|_{L^2_x} \\
& \leq &  C\|\langle x\rangle^{\al} F\|_{L^1_x L^2_t},
\end{eqnarray*}
where in the last step, we have used the dual estimate to
\eqref{a:10}. In order to control the $U$ term, we observe
that the set of all functions
$\{g_1(t) g_2(x): g_1\in L^2_t, g_2\in L^1_x\}$ is dense in
$L^1_x L^2_t$. The estimate that we need follows from
$$
\left\| \langle x\rangle^{-\al} \int_{-\infty}^\infty e^{-i t \la}
\check{g_1}(\la)
 \left[ R^+_{\ch}(\la) + R^-_{\ch}(\la) \right]
 g_2 d\la \right\|_{L^\infty_x L^2_t}\leq C \| g_1 \|_{L^2_t}
 \| \langle x \rangle^{\al} g_2\|_{L^1_x}.
$$
The left-hand side is controlled by Minkowski's inequality and Plancherel's theorem
in the time variable,
\begin{eqnarray*}
\|\langle x\rangle^{-\al} \|\check{g_1}(\la) \left[ R^+_{\ch}(\la)
+ R^-_{\ch}(\la) \right] g_2 \|_{L^2_\la}\|_{L^\infty_x} \leq C
\|\check{g_1} (\la)\|_{L^2_\la} \sup_{\la \in \R} \|
R^{\pm}_{\ch}(\la) \|_{L^2_{\al} \mapsto L^2_{-\al}} \| \langle x
\rangle^{\al} g_2 \|_{L^\infty_x}.
\end{eqnarray*}
Using bound \eqref{b:20} of Proposition \ref{prop:a1} for any
$\alpha \geq 1$, we bound the last expression by $ C
\|g_1\|_{L^2_t} \| \langle x \rangle g_2\|_{L^1_x}, $ which
completes the proof of estimate \eqref{b:15}.
\end{proof}

\subsubsection{Proof of \eqref{a:11}}

Using the Cauchy formula (\ref{Cauchy-formula}) for
$$
g_{x,t}(\lambda) := (1-\chi_M(\la)) e^{-i t \ch} P_{a.c.}(\ch) f,
$$
we can see that for each fixed value of $x$, this function is a
multiple of the Fourier transform of the function
$$
g_x(\la) := (1-\chi_M(\la)) \left( \left[ R^+_{\ch}(\la) -
R^-_{\ch}(\la) \right] f \right)(x),
$$
evaluated at $t$. Therefore, by Plancherel's theorem, we have
$$
\| (1-\chi_M(\la)) e^{-i t \ch} P_{a.c.}(\ch) f\|_{L^2_t}= C
\|g_x\|_{L^2_{\lambda}}.
$$

It is sufficient to control
$$
\sup_{x \in \R} \| (1-\chi_M(\la)) R^{\pm}_{\ch}(\la) f(x)
\|_{L^2_\la}\leq C \|f\|_{L^2_x},
$$
which we will do next. By  iterating the resolvent identities,
$$
R_{\ch}=R_0-R_{\ch} V R_0=R_0-R_0 V R_{\ch},
$$
we get the representation formula
\begin{equation}
\label{a:35} R_{\ch} = R_0 - R_{\ch}  V R_0 = R_0-R_0 V R_0+R_0 V
R_{\ch} V R_0.
\end{equation}
where $R_0$ is the resolvent of the free Dirac operator $D$
defined by \eqref{a:30}. For the first term, we have
$$
\sup_{x \in \R} \| (1-\chi_M(\la)) R^{\pm}_0(\la)
f(x)\|_{L^2_\la}.
$$
By symmetry, it suffices to consider only positive values of $\la$,
whence we need to control
$$
\sup_{x \in \R} \int_{M}^\infty |R_0^{\pm}(\la) f(x)|^2d\la.
$$

We compute the resolvent $R_0(\mu)$
$$
R_0(\mu)=(D-\mu)^{-1} = (1-\p_x^2-\mu^2)^{-1} \left(
\begin{array}{c c} -i \p_x+\mu & -1 \\ -1 & i\p_x +\mu \end{array} \right)
$$
for $\mu\notin \si(D)=[-\infty, -1]\cup [1,\infty]$.  By analytic
continuation, we may define the resolvent even for values on the
spectrum of $\si(D)$. Since we need such a formula for values of
$\mu\in (M,\infty)$, it is convenient to introduce  a change of
variables $\mu=\sqrt{k^2+1}$. Note that
$d\mu=k(k^2+1)^{-1/2}dk\sim dk$ and the interval of integration
becomes $(\sqrt{M^2-1},\infty)$. Now, since the resolvent
operator $(-\p_x^2-k^2\pm 0)^{-1}$ is given by a convolution with
the explicit  kernel $\f{e^{\pm i k |\cdot|}}{2i k}$, it is clear
that $R_0^{\pm}(\mu) f$ is a linear combination of convolution
operators with kernels
 \begin{equation}
 \label{kernels}
 e^{  \pm i  k |x|} {\rm sgn}(x) , \ \f{e^{ \pm i k |x|}}{k}, \ \f{e^{  \pm i k |x|}\sqrt{k^2+1}}{k}.
\end{equation}

We shall consider the first type of operators, the second one has a stronger decay,
while the third one is basically the same as the first one. By Plancherel's theorem
applied to the functions $f(y)\chi_{y<x}$ and $f(y)\chi_{y>x}$, we have
\begin{eqnarray*}
& & \int_{\sqrt{M^2-1}}^\infty \left| \int_{-\infty}^\infty e^{\pm i k |x-y|} {\rm sgn}(x-y)f(y)dy
\right|^2 dk \\
& & \leq 2 \int_{\sqrt{M^2-1}}^\infty \left( \left|
\int_{-\infty}^x e^{\mp i k y}f(y)dy \right|^2 + \left|
\int_{x}^\infty e^{\pm i k y}f(y)dy \right|^2 \right) dk \leq C
\|f\|_{L^2_x}^2.
\end{eqnarray*}

Similarly, we estimate the contribution of the  second term $R_0 V
R_0$ in the expansion (\ref{a:35}). Again, we have to deal with
different terms of the convolution operators, but the hardest one
is again $e^{i k |x|} {\rm sgn}(x)$. We get
\begin{eqnarray*}
& &   \int_{\sqrt{M^2-1}}^\infty \left|
\int_{-\infty}^\infty e^{\pm i k |x-y|} {\rm sgn}(x-y)V(y) \int e^{\pm ik|y-z|} {\rm sgn}(y-z)f(z)dz dy
\right|^2 dk \\
  & & \leq C \|V\|_{L^1_x}^2 \sup_{y \in \R}
 \int_{\sqrt{M^2-1}}^\infty \left| \int e^{\pm ik|y-z|} {\rm sgn}(y-z)f(z)dz \right|^2 dk
 \leq C \|V\|_{L^1_x}^2 \|f\|_{L^2_x}^2,
\end{eqnarray*}
where in the first inequality, we have applied Minkowski's and at
the second inequality, we have applied our previous estimate.

In order to estimate the last term in \eqref{a:35}, we use bound (\ref{b:30}) of
Proposition \ref{prop:a1} and get
\begin{eqnarray*}
& & \int_{\sqrt{M^2-1}}^\infty \left| \int e^{\pm i k |x-y|} {\rm
sgn}(x-y) V(y)[R^{\pm}_{\ch}(\sqrt{1+k^2}) V (R^{\pm}_0(\sqrt{1+k^2}) f)(y) dy \right|^2 dk \\
& & \leq C \|<x>^{\al} V\|_{L^2_x}^2 \int_{\sqrt{M^2-1}}^\infty
\left\| <y>^{-\al} R^{\pm}_{\ch}(\sqrt{1+k^2}) V
R^{\pm}_0(\sqrt{1+k^2}) f \right\|_{L^2_y}^2 dk \\
& & \leq C \|<x>^{\al} V\|_{L^2_x}^2 \sup_{y \in \R}
\int_{\sqrt{M^2-1}}^\infty |R^{\pm}_0(\sqrt{1+k^2}) f(y)|^2 dk \\
& & \leq C \|<x>^{\al} V\|_{L^2_x}^2 \|f\|_{L^2_x}^2.
\end{eqnarray*}
This concludes the proof of \eqref{a:11}.

\subsubsection{Proof of \eqref{a:12}}

We shall prove that
\begin{equation}
\label{bound-auxillary} \sup_{x \in \R} <x>^{-3/2} \| \chi_M(\la)
(R_{\ch}^{\pm}(\la) f)(x) \|_{L^2_{\lambda}} \leq C \| f
\|_{L^2_x},
\end{equation}
which implies bound (\ref{a:12}) by Plancherels' theorem and
Cauchy's formula (\ref{Cauchy-formula}). To prove
(\ref{bound-auxillary}) for $\la \leq -1$, we use representation
(\ref{resolvent-limiting}) and write explicitly
\begin{eqnarray*}
\| \chi_M(\la) (R_{\ch}^+(\la) f)(x) \|^2_{L^2_{\lambda}} =
\int_{-\infty}^{-1} \chi_M^2(\la) \left| (R_{\ch}^+(\la) f)(x)
\right|^2 d \la = \int_{-\sqrt{M^2 - 1}}^{0} \frac{| \tilde{\bf
f}(x,k)|^2 |k| dk}{4 k^2 |a^+(k)|^2 \sqrt{1 + k^2}},
\end{eqnarray*}
where
$$
\tilde{\bf f}(x,k) :=  {\bf u}^+(x;k) \int_{-\infty}^x [\sigma_1
{\bf u}^-(y;k)]^T f(y) dy + {\bf u}^-(x;k) \int_{x}^{\infty}
[\sigma_1 {\bf u}^+(y;k)]^T f(y) dy.
$$

For definiteness, let us assume that $x \geq 0$. We represent
\begin{eqnarray*}
\int_{-\infty}^x [\sigma_1 {\bf u}^-(y;k)]^T f(y) dy & = &
\int_{0}^x [\sigma_1 {\bf u}^-(y;k)]^T f(y) dy +
\int_{-\infty}^0 [\alpha_-,1] f(y) e^{-iky} dy \\
& \phantom{t} &  + \int_{-\infty}^0 \left( [\sigma_1 {\bf
m}^-(y;k)]^T - [\alpha_-,1] \right) f(y) e^{-iky} dy \equiv I_1 +
I_2 + I_3
\end{eqnarray*}
and
\begin{eqnarray*}
\int_{x}^{\infty} [\sigma_1 {\bf u}^+(y;k)]^T f(y) dy & = &
\int_{x}^{\infty} [\alpha_+,1] f(y) e^{iky} dy + \int_{x}^{\infty}
\left( [\sigma_1 {\bf m}^+(y;k)]^T - [\alpha_+,1] \right) f(y)
e^{iky} dy \\ & \phantom{t} & \equiv I_4 + I_5.
\end{eqnarray*}

Using Proposition \ref{proposition-Jost-functions} and
Cauchy--Schwarz inequality, we have
\begin{eqnarray*}
|I_1| & \leq & \| {\bf u}^-(\cdot;k) \|_{L^2_x(0,x)} \| f
\|_{L^2_x} \leq C \langle x \rangle^{3/2} \| f \|_{L^2_x}, \\
|I_3| & \leq & \| {\bf m}^-(\cdot;k) - [1,\alpha_-]^T
\|_{L^2_x(\R_-)}
\| f \|_{L^2_x} \leq C \| \langle x \rangle^3 V \|_{L^\infty_x} \| f \|_{L^2_x}, \\
|I_5| & \leq & \| {\bf m}^+(\cdot;k) - [1,\alpha_+]^T
\|_{L^2_x(\R_+)} \| f \|_{L^2_x} \leq C \| \langle x \rangle^3 V
\|_{L^\infty_x} \| f \|_{L^2_x}.
\end{eqnarray*}
The estimates for $I_3$ and $I_5$ follow from the bound
\begin{equation}
\label{bound-I-3-I-5}
\| {\bf m}^+(\cdot;k) - [1,\alpha_+]^T \|_{L^2_x(\R_+)} + \| {\bf
m}^-(\cdot;k) - [1,\alpha_-]^T \|_{L^2_x(\R_-)} \leq C
\| \langle x \rangle^3 V \|_{L^\infty_x}
\end{equation}
which we prove now. We need only control the first term, the other one is
controlled in a similar matter.

By the formula \eqref{matrix-potential}, for all $x \in \R$ and all $k \in \R$ near $k = 0$, there is $C  > 0$
such that
$$
|G^+(x;k)| \leq C \langle x \rangle.
$$
By Proposition \ref{proposition-Jost-functions}, for all $x > 0$, there is $C > 0$ such that
$$
|{\bf m}^+(x,k)| = |{\bf u}^+(x,k)| \leq C.
$$
Thus, by the integral equation \eqref{integral-equations-Volterra}, we get for all $x>0$,
\begin{eqnarray*}
 | {\bf m}^+(x;k) - [1,\alpha_+]^T| &\leq & C \int_x^\infty
\langle x-y\rangle |V(y) dy \leq C
\| \langle x \rangle^3 V \|_{L^\infty_x}
\int_0^\infty
\langle z\rangle \f{1}{\langle x+z \rangle^3}  dz \\
&\leq & C \| \langle x \rangle^3 V \|_{L^\infty_x} <x>^{-1}.
\end{eqnarray*}
This computation completes the proof of the first inequality in (\ref{bound-I-3-I-5}).

On the other hand, for any finite $M > 1$, Plancherel's theorem gives
$$
\int_{-\sqrt{M^2-1}}^0 \left( |I_2|^2 + |I_4|^2 \right) dk \leq C
\| f \|^2_{L^2_x}.
$$
Since $k a^+(k)$ is bounded away from zero as $k \to 0$ and
$|a^+(k)| \geq 1$, we obtain
$$
\int_{-\sqrt{M^2 - 1}}^{0} \frac{| \tilde{\bf f}(x,k)|^2 |k| dk}{4
k^2 |a^+(k)|^2 \sqrt{1 + k^2}} \leq C (1 + \langle x \rangle^3) \|
f \|^2_{L^2_x},
$$
which concludes the proof of bound (\ref{bound-auxillary}) and
hence of bound (\ref{a:12}).

\subsection{Strichartz estimates}

We use the following standard definition.
\begin{definition}
\label{defi:1} We say that a pair $(q,r)$ is Strichartz admissible
for the nonlinear Dirac equations if
$$
q \geq 2, \quad r\geq 2 \quad \mbox{\rm and} \quad \frac{2}{q} + \frac{1}{r} \leq \frac{1}{2}.
$$
In particular, $(q,r)=(4,\infty)$ and $(q,r)=(\infty,2)$ are
end-point Strichartz pairs.
\end{definition}

\begin{lemma}
\label{le:1} Let $(q,r)$ be a Strichartz admissible pair, $s\geq
0$, and $\ve>0$. Then, there are constants $C_\ve > 0$ and $C > 0$ such that
 \begin{eqnarray}
 \label{120}
 & &   \|e^{-i t \ch} P_{a.c.}(\ch)f\|_{L^{4}_t L^\infty_x}\leq
 C_\ve \|f\|_{H_x^{3/4+\ve}},  \\
\label{a50}
 & &   \|e^{-i t \ch} P_{a.c.}(\ch)f\|_{L^{\infty}_t H^s_x}\leq
 C  \|f\|_{H^s_x},\\
 \label{a:3}
 & &  \left\|\int_0^t e^{-i(t-\tau) \ch} P_{a.c.}(\ch)
 F(\tau, \cdot)d\tau \right\|_{L^{\infty}_t H^1_x\cap L^q_t L^r_x}\leq C
 \|F\|_{L^{1}_t H^1_x}.
 \end{eqnarray}
\end{lemma}

\begin{proof}
Let us first comment on the estimates \eqref{a50} and \eqref{a:3}.
It is easy to see by the self-adjointness of $\ch$ that
\eqref{a50} is trivial for $s=0$. We easily extend to all integer
values of $s$ by the observation that $\p_x$ behaves like $\ch$
and commuting $\ch$ with $e^{-i t \ch} P_{a.c.}(\ch)$. This is
made precise in formula \eqref{derivative-via-H} below. We then
conclude by interpolation to obtain \eqref{a50} for all
nonnegative values of $s$. Regarding \eqref{a:3}, it follows by an
easy application of Lemma \ref{christ_kiselev} combined with the
dual estimate of \eqref{120}.

Thus, it remains to show \eqref{120}. We will in fact deduce this
Strichartz estimate for the perturbed Dirac operator $\ch$ by
using the corresponding result for the free Dirac operator $D$, in
addition to the weighted estimates in Lemma \ref{le:2}. This is in
essence the approach taken by Rodnianski and Schlag, \cite{RS}.
Let us first record the Strichartz estimates for the Dirac
operator $D$
\begin{equation}
\label{c:2}
 \|e^{-i t D} f\|_{L^{q}_t L^r_x}\leq C_\de \|f\|_{H_x^{s(q,r)}},
 \quad s(q,r)=\f{1}{2}+\f{1}{q}-\f{1}{r}.
\end{equation}
for all Strichartz admissible pairs $(q,r)$, so that $q\geq
4+\de$. This of course looks exactly the same as the estimates
that one gets from interpolating between \eqref{120} and
\eqref{a50}. We refer the reader to recent work of Nakamura-Ozawa,
\cite{NO} (more specifically Lemma 2.1 with $\theta=1, \la=3/2,
n=1$) for a reference for this result. Note that this result would
not extend to the full range $q=4, r=\infty$, unless we are
willing to replace the $L^\infty$ by  the Besov space
$B^0_{\infty,2}$ (which we are avoiding for the purpose of
simplicity). In order to extend this to the useful endpoint $q=4,
r=\infty$, we must introduce slight loss of smoothness, so we have
\begin{equation}
\label{c:5} \|e^{-i t D} f\|_{L^{4}_t L^\infty_x}\leq C_\ve
\|f\|_{H_x^{3/4+\ve}}.
\end{equation}

Fix now $\ve>0$ and take a test function $f=P_{a.c.}(\ch) f \in H^{3/4+\ve}$. Recall that since
$\ch= D + V(x)$, we may write
$$
e^{-i t \ch} f = e^{-i t D}f - i \int_0^t e^{-i (t-s) D} V e^{-i s \ch} f ds.
$$
Furthermore, we may write the symmetric matrix $V(x)$ as the
product of $V_1(x)$ and $V_2(x)$, where both $V_1(x)$ and $V_2(x)$
are $C^1$-smooth and have fast decay at spatial infinity. For
instance, one may pick $V_1(x)= V(x) \langle x\rangle^{10}$ and
$V_2(x) = \langle x\rangle^{-10} Id$. We have
\begin{eqnarray*}
\|e^{-i t \ch} P_{a.c.}(\ch)  f\|_{L^4_t L^\infty_x} & \leq &
 \|e^{-i t D} f\|_{L^4_t L^\infty_x}+ \left\| \int_0^t e^{-i(t-s) D}
 V_1 V_2 e^{-i s \ch} P_{a.c.}(\ch)  f ds \right\|_{L^4_t L^\infty_x} \\
 & \leq & C_\ve \|f\|_{H_x^{3/4+\ve}}+ \left\| \int_0^t e^{-i(t-s)D}
 V_1 V_2 e^{-i s \ch} P_{a.c.}(\ch)  f ds \right\|_{L^4_t L^\infty_x}.
\end{eqnarray*}

At this stage, in order to estimate the second term, we will use
Lemma \ref{christ_kiselev}. Let
$K(t,s)=e^{-i (t-s) D} V_1$ be considered as acting between $L^2_t
H^{3/4+\ve}_x$ to $L^4_t L^\infty_x$. The Duhamel's term that we
need to estimate is
$$
M(t) = \int_0^t K(t,s) V_2 e^{-i s \ch} P_{a.c.}(\ch)  f
ds = \tilde{K} V_2 e^{-i t \ch}P_{a.c.}(\ch)  f,
$$
where $\tilde{K}$ is defined by (\ref{definition-tilde-K}). It
follows from Lemma \ref{christ_kiselev} (since $q=4>2=p$, this
lemma can be applied) that
$$
\| M \|_{L^4_t L^\infty_x}\leq C \|K\|_{L^2_t H^{3/4+\ve}_x \to L^4_t L^\infty_x}
\|V_2 e^{-i t \ch}P_{a.c.}(\ch)  f\|_{L^2_t H^{3/4+\ve}_x}.
$$
We need estimate then the operator norm $\|K\|_{L^2_t  H^{3/4+\ve}_x \to
L^4_t L^\infty_x}$. We have by \eqref{c:5}
\begin{eqnarray*}
\|K G\|_{L^4_t L^\infty_x} &= & \left\| e^{-i t D}
\int_{-\infty}^\infty e^{i s D} V_1 G(s, \cdot) ds \right\|_{L^4_t
L^\infty_x} \leq C_\ve \left\| \int_{-\infty}^\infty e^{i s D} V_1
G(s, \cdot) ds \right\|_{H_x^{3/4+\ve}}.
\end{eqnarray*}

We will show that for $0\leq s\leq 1$,
\begin{eqnarray}
\label{c:10} \left\|\int_{-\infty}^\infty e^{i s D} V_1 G(s,
\cdot) ds \right\|_{H^s_x} \leq C_{s,V_1} \|G\|_{L^2_t H^s_x},
\end{eqnarray}
and
\begin{eqnarray}
\label{c:15} \|V_2 e^{-i t \ch}P_{a.c.}(\ch)  f\|_{L^2_t
H^{s}_x}\leq C_{V_2}\|f\|_{H^s_x},
\end{eqnarray}
which implies what is needed. Indeed, for $s=3/4+\ve$, we deduce
$$
\|M\|_{L^4_t L^\infty_x}\leq C_{V_1} \|V_2 e^{-i t \ch}P_{a.c.}(\ch)  f\|_{L^2_t H^{3/4+\ve}_x}\leq
C_{V_1, V_2}\|f\|_{H^{3/4+\ve}_x}.
$$
It thus suffices to establish \eqref{c:10} and \eqref{c:15}. By
interpolation, it suffices to check both only for $s=0$ and $s=1$.
The statements for  $s=0$ in fact follow from the corresponding
arguments for $s=1$, so we concentrate on $s=1$. For \eqref{c:10},
(observe that $\p_x e^{i t D}=e^{i t D} \p_x$), we have
\begin{eqnarray*}
& & \left\| \int_{-\infty}^\infty e^{i s D} V_1 G(s, \cdot) ds
\right\|_{H^1_x}\leq \left\| \int_{-\infty}^\infty e^{i s D} V_1
G(s, \cdot) ds \right\|_{L^2_x}+ \left\| \int_{-\infty}^\infty
e^{i s D} \p_x[V_1 G(s, \cdot)] ds \right\|_{L^2_x}.
\end{eqnarray*}
By the dual estimate to \eqref{c:2} (recall $s(\infty,2)=0$),
the right-hand side of the last inequality is estimated by
\begin{eqnarray*}
& &C(\|V_1 G(s, \cdot)\|_{L^1_t L^2_x}+\|\p_x[V_1 G(s, \cdot)]\|_{L^1_t L^2_x})\leq C (\|V_1\|_{L^\infty_x}+\|V'_1\|_{L^\infty_x})\|G\|_{L^1_t H^1_x}.
\end{eqnarray*}
This is the proof of \eqref{c:10}.

Next, we need to deal with derivatives in the estimates for the
perturbed evolution. From the formula $D=\ch-V(x)$,  we have the
equivalence
\begin{equation}
\label{derivative-via-H}
\left\|\left(\begin{array}{c}u \\ v\end{array}\right)\right\|_{H^1}\sim \left\|\ch \left(\begin{array}{c}u \\ v\end{array}\right)\right\|_{L^2}+\left\|\left(\begin{array}{c}u \\ v\end{array}\right)\right\|_{L^2}
\end{equation}
which will be used repeatedly in the arguments to follow.
Regarding \eqref{c:15} for $s=1$, we use \eqref{derivative-via-H}
to obtain
\begin{eqnarray*}
\|\p_x[V_2 e^{-i t \ch}f]\|_{L^2_t L^2_x} \leq \|V_2' e^{-i t
\ch}f\|_{L^2_t L^2_x}+ \|V_2 e^{-i t \ch}\ch f\|_{L^2_t L^2_x} +
\|V_2 e^{-i t \ch} f\|_{L^2_t L^2_x}.
\end{eqnarray*}
Now, since $|V_2'(x)|+|V_2(x)|\leq \langle x\rangle^{-10}$, we estimate the last three quantities by
\begin{eqnarray*}
C \|\langle x \rangle^{-5}\|_{L^2_x}(\|\langle x\rangle^{-5} e^{-i
t \ch} f\|_{L^\infty_x L^2_t}+\|\langle x\rangle^{-5} e^{-i t \ch}
 \ch f\|_{L^\infty_x L^2_t}+ \|\langle x\rangle^{-5} e^{-i t \ch} f\|_{L^\infty_x L^2_t}) \\
<leq C(\|f\|_{L^2_x}+\|\ch f\|_{L^2_x}) \leq C\|f\|_{H^1_x},
\end{eqnarray*}
where bound \eqref{a:10} and H\"older's inequality are used. This computation
establishes \eqref{c:15} and hence Lemma \ref{le:1}.
\end{proof}

\subsection{Additional estimates}

Mizumachi estimates and Strichartz estimates admit a number of useful corollaries.

\begin{corollary}
\label{corStr}
Let $(q,r)$ be and admissible Strichartz pair such that $q\geq 4+\de$. For each $\de>0$,
there is $C_{\de} > 0$ such that
\begin{equation}
\label{c:1} \left\| \int_0^t e^{i\tau  \ch} P_{a.c.}(\ch)F(\tau,
\cdot)d\tau \right\|_{L^2_x}\leq
 C_{q,\de}   \|F\|_{L^{q'}_t W^{s(q,r),r'}_x}, \quad
 s(q,r)=\f{1}{2}+\f{1}{q}-\f{1}{r},
\end{equation}
where $(q',r')$ are duals of $(q,r)$.
\end{corollary}

\begin{proof}
The result is obtained from \eqref{120} and \eqref{a50} by duality and interpolation.
\end{proof}

\begin{corollary}
\label{mizStr}
Fix $\alpha > 2$. There is $C_{\al} > 0$ such that
\begin{eqnarray}
\label{a:90}
 & & \left\| \int_0^t e^{-i(t-\tau) \ch} P_{a.c.}(\ch)
 F(\tau, \cdot)d\tau \right\|_{L^{\infty}_t H^1_x \cap L^4_t L^\infty_x} \leq
 C_{\al} (\|\langle x\rangle^\al F\|_{L^1_x L^2_t}+
 \|\langle x\rangle^\al \p_x F\|_{L^1_x L^2_t}).
\end{eqnarray}
\end{corollary}

\begin{proof}
Due to the density of $\{g_1(t) g_2(x): g_1\in L^2_t, g_2\in L^1_x\}$   in
 $L^1_x L^2_t$, it will suffice to show
 \begin{equation}
\label{a:100}
 \left\| \int_0^t e^{-i(t-\tau) \ch} P_{a.c.}(\ch)
 g_1(\tau)g_2 d\tau \right\|_{L^{\infty}_t H^1_x\cap L^4_t L^\infty_x}\leq
 C\|g_1\|_{L^2_t}(\|\langle x\rangle^\al g_2\|_{L^1_x}+
 \|\langle x\rangle^\al \p_x g_2\|_{L^1_x})
 \end{equation}
By Lemma \ref{christ_kiselev}, we need to show that
$$
\left\| \int_{-\infty}^\infty e^{-i(t-\tau) \ch} P_{a.c.}(\ch)
 g_1(\tau)g_2 d\tau \right\|_{L^{\infty}_t H^1_x\cap L^4_t L^\infty_x}\leq
 C\|g_1\|_{L^2_t}(\|\langle x\rangle^\al g_2\|_{L^1_x}+
 \|\langle x\rangle^\al \p_x g_2\|_{L^1_x})
$$
By \eqref{120} and \eqref{a50}, we have
\begin{eqnarray*}
& & \left\| e^{-it \ch}\int_{-\infty}^\infty e^{i \tau \ch}
P_{a.c.}(\ch) g_1(\tau)g_2 d\tau \right\|_{L^{\infty}_t H^1_x\cap
L^4_t L^\infty_x}\leq \left\|\int_{-\infty}^\infty e^{i \tau \ch}
P_{a.c.}(\ch) g_1(\tau)g_2 d\tau \right\|_{H^1_x},
\end{eqnarray*}
Again, one may convert one derivative to $\ch-V$ by the
equivalence (\ref{derivative-via-H}), whence we further estimate
by the dual of \eqref{a:10},
\begin{eqnarray*}
& & \left\| \int_{-\infty}^\infty e^{i \tau \ch} P_{a.c.}(\ch)
 g_1(\tau) g_2 d\tau \right\|_{H^1_x} \\
& & \leq C \left( \left\|\int_{-\infty}^\infty e^{i \tau \ch}
P_{a.c.}(\ch) g_1(\tau)\ch g_2(\cdot) d\tau \right\|_{L^2_x} +
\left\| V \int_{-\infty}^\infty e^{i \tau \ch} P_{a.c.}(\ch)
 g_1(\tau)  g_2(\cdot) d\tau \right\|_{L^2_x} \right) \\
& & \leq C\|g_1\|_{L^2_t}(\|\langle x\rangle^{\al} \ch
g_2\|_{L^1_x}+\|\langle x\rangle^{\al}  g_2\|_{L^1_x}) \\
& & \leq C\|g_1\|_{L^2_t}(\|\langle x\rangle^\al g_2\|_{L^1_x}+
 \|\langle x\rangle^\al \p_x g_2\|_{L^1_x}),
\end{eqnarray*}
which is the desired estimate.
\end{proof}

\begin{corollary}
\label{cor-1}
Fix $\al>2$. There is $C_{\al} > 0$ such that
\begin{eqnarray}
\label{a:55} & &  \left\|\langle x\rangle^{-\al} \int_0^t
e^{-i(t-\tau)\ch} P_{a.c.}(\ch)F(\tau, \cdot)d\tau
\right\|_{L^\infty_x L^2_t}\leq C \|F\|_{L^1_t L^2_x},
 \end{eqnarray}
More generally, let $(q,r)$ be an admissible Strichartz pair. Then,
\begin{equation}
\label{a:80} \left\|\langle x\rangle^{-\al} \int_0^t e^{-i(t-\tau)
\ch} P_{a.c.}(\ch) F(\tau, \cdot) d\tau \right\|_{L^\infty_x
L^2_t}\leq C_{\al} \|F\|_{L^{q'}_{t} W_x^{1,r'}},
\end{equation}
where $(q',r')$ is a dual pair.
\end{corollary}

\begin{proof}
The proof of \eqref{a:55} is by averaging the estimate
\eqref{a:10}. More precisely, using the triangle inequality  and
estimate \eqref{a:10} yields
\begin{eqnarray*}
\left\|\langle x\rangle^{-\al} \int_0^t e^{-i(t-\tau)\ch}
P_{a.c.}(\ch)F(\tau, \cdot)d\tau \right\|_{L^\infty_x L^2_t} &
\leq & C  \int_{-\infty}^\infty \|\langle x\rangle^{-\al}e^{-i
(t-\tau) \ch} P_{a.c.}(\ch)F(\tau, \cdot)\|_{L^\infty_x L^2_t} d\tau \\
& \leq &  C \int_{-\infty}^\infty \|F(\tau, \cdot)\|_{L^2_x}
d\tau= C\|F\|_{L^1_t L^2_x}.
\end{eqnarray*}
For the proof of \eqref{a:80}, we use Lemma \ref{christ_kiselev}.
It will suffice to bound the operator
$$
T F(t)=\langle x_0\rangle^{-\al} \int_{-\infty}^\infty
e^{-i(t-\tau) \ch} F(\tau, \cdot) d\tau \biggr|_{x=x_0} : L^{q'}_t
W_x^{1,r'}\to L^2_t
$$
for any fixed $x_0\in \R$. We have, by \eqref{a:10}
\begin{eqnarray*}
& & \|T F\|_{L^2_t}\leq \left\| \langle x\rangle^{-\al} e^{-i t
\ch} \int_{-\infty}^\infty e^{i \tau \ch} F(\tau, \cdot) d\tau
\right\|_{L^\infty_x L^2_t}\leq C \left\|\int_{-\infty}^\infty
e^{i \tau \ch} F(\tau, \cdot) d\tau \right\|_{L^2_x}.
\end{eqnarray*}
By Corollary \ref{corStr}, we bound the last expression by
$$
C\|F\|_{L^{q'}_t W^{\f{3}{2q}+\de,r'}_x} \leq C
\|F\|_{L^{q'}_t W^{1 ,r'}_x},
$$
as stated in \eqref{a:80}. In the last step, we have used that if
$4\leq q\leq \infty$ and $\de \ll 1$, then $\f{3}{2q}+\de<1$.
\end{proof}

\begin{corollary}
\label{le:15}
Fix $\al>2$. There is $C_{\al} > 0$ such that
\begin{equation}
\label{a:70}
\|\langle x\rangle^{-\al} \p_x e^{-i t \ch} P_{a.c.}(\ch) f\|_{L^\infty_x L^2_t }\leq C_{\al} \|f\|_{H^1_x},
\end{equation}
\begin{equation}
\label{a:65} \left\| \langle x\rangle^{-\al} \p_x \int_0^t
e^{-i(t-\tau)\ch} P_{a.c.}(\ch)F(\tau, \cdot)d\tau
\right\|_{L^\infty_x L^2_t}\leq C \|F\|_{L^1_t H^1_x},
\end{equation}
and
\begin{equation}
\label{a:85} \left\|\langle x\rangle^{-\al} \p_x \int_0^t
e^{-i(t-\tau)\ch} P_{a.c.}(\ch)F(\tau, \cdot)d\tau
\right\|_{L^\infty_x L^2_t} \leq C(\| \langle x\rangle^{\al}
F\|_{L^\infty_x L^2_t} + \| \langle x\rangle^{\al} \p_x
F\|_{L^\infty_x L^2_t}).
\end{equation}
\end{corollary}

\begin{proof}
The proof of the estimate  \eqref{a:70} is based again on the
equivalence (\ref{derivative-via-H}). Since $\ch$ commutes with
all functions of $\ch$ (by the functional calculus),  we have from
\eqref{derivative-via-H} and \eqref{a:10}
\begin{eqnarray*}
\|\langle x\rangle^{-\al} \p_x e^{-i t \ch} P_{a.c.} (\ch) f\|_{L^\infty_x L^2_t} & \leq &
\|\langle x\rangle^{-\al} e^{-i t \ch} P_{a.c.}  (\ch) \ch f\|_{L^\infty_x L^2_t}
+ \|\langle x\rangle^{-\al} V e^{-i t \ch} P_{a.c.}   f\|_{L^\infty_x L^2_t} \\
& \leq &  C (\|\ch f\|_{L^2_x}+
\|V\|_{L^\infty} \|f\|_{L^2_x})\leq C_V \|f\|_{H^1_x}.
\end{eqnarray*}
The proof of the estimate \eqref{a:65} is by averaging. Indeed, by \eqref{a:70},
\begin{eqnarray*}
& & \left\| \langle x\rangle^{-\al} \p_x \int_0^t
e^{-i(t-\tau)\ch} P_{a.c.}(\ch)F(\tau, \cdot)d\tau
\right\|_{L^\infty_x L^2_t} \\
& & \leq C \int_{-\infty}^\infty \|\langle x\rangle^{-\al} \p_x
e^{-i(t-\tau)\ch} P_{a.c.}(\ch)F(\tau, \cdot)\|_{L^\infty_x L^2_t} d\tau \\
&  & \leq C \int_{-\infty}^\infty \|F(\tau, \cdot)\|_{H^1_x} d\tau
= C\|F\|_{L^1_t H^1_x}
\end{eqnarray*}

For the proof of the estimate \eqref{a:85}, we apply again the
equivalence (\ref{derivative-via-H}) and then we use \eqref{b:15}.
We have
\begin{eqnarray*}
& & \left\|\langle x\rangle^{-\al} \p_x \int_0^t e^{-i(t-\tau)\ch}
  P_{a.c.}(\ch)F(\tau, \cdot)d\tau \right\|_{L^\infty_x L^2_t} \\
& & \leq \left\|\langle x\rangle^{-\al}   \int_0^t
e^{-i(t-\tau)\ch} P_{a.c.}(\ch) \ch F(\tau, \cdot)d\tau \right\|_{L^\infty_x L^2_t}
+  \left\|\langle x\rangle^{-\al} V(x) \int_0^t
e^{-i(t-\tau)\ch} P_{a.c.}(\ch)  F(\tau, \cdot)d\tau \right\|_{L^\infty_x L^2_t} \\
& & \leq C( \|\langle x\rangle^{\al} \ch F\|_{L^1_x L^2_t} + \|\langle x\rangle^{\al} F\|_{L^1_x L^2_t}) \\
& & \leq C( \|\langle x\rangle^{\al} F\|_{L^1_x L^2_t}+  \|\langle x\rangle^{\al} \p_x F\|_{L^1_x L^2_t}).
\end{eqnarray*}
This concludes the proof of the corollary.
\end{proof}

\section{Proof of the main theorem}

We first formulate the solution and the nonlinearity spaces. Let
\begin{eqnarray*}
\| {\bf Y} \|_{X_1} := \| {\bf Y} \|_{L^4_t L^\infty_x}+  \| {\bf Y} \|_{L^\infty_t H^1_x}, \quad
\| {\bf Y} \|_{X_2}:= \|\langle x \rangle^{-\al} {\bf Y} \|_{L^\infty_x L^2_t}+
\|\langle x\rangle^{-\al} \p_x {\bf Y} \|_{L^\infty_x L^2_t},
\end{eqnarray*}
and $\|{\bf Y} \|_X:=\|{\bf Y}\|_{X_1}+\|{\bf Y}\|_{X_2}$. The nonlinearity space
is defined via the norm
$$
\| {\bf F} \|_{\cn}:=\inf_{{\bf F} = {\bf F}_1 + {\bf F}_2} \|{\bf F}_1\|_{L^1_t H^1_x}+
(\|\langle x\rangle^\al {\bf F}_2\|_{L^1_x L^2_t}+
\|\langle x\rangle^\al \p_x {\bf F}_2\|_{L^1_x L^2_t}).
$$

Consider the Cauchy problem for the inhomogeneous linear equation,
projected along the absolutely continuous spectrum of $\ch$
\begin{equation}
\label{inhomogen-system}
\left\{ \begin{array}{l} i \frac{d {\bf Y}}{d t} = \ch {\bf Y} + P_{a.c.}(\ch) {\bf F}, \\
{\bf Y}(0) = {\bf Y}_0 = P_{a.c.}(\ch) {\bf Y}_0.
\end{array} \right.
\end{equation}
When one interprets correctly the results of the dispersive decay estimates 
(Section 7) in the notations above,
we get that a solution to the Cauchy problem (\ref{inhomogen-system}) satisfies
\begin{equation}
\label{e:2} \| {\bf Y} \|_{X} \leq C(\| {\bf Y}_0 \|_{H^1} + \|
{\bf F} \|_{\cn}).
\end{equation}
For the proof of the main theorem, we need to show the existence
of small solutions for the  system of two (scalar) ordinary
differential equations \eqref{projection-equations} for
$\om$ and $\theta$ coupled with the
partial differential equation \eqref{residual-equation} for ${\bf Y}$.

Since the right-hand side of equation (\ref{residual-equation}) is not projected
to the continuous spectrum of $\ch$, we decompose
\begin{equation}\label{decomposition-last}
{\bf Y} = a {\bf u}_0 + {\bf Z}, \quad a = \langle {\bf u}_0, {\bf Y} \rangle_{L^2},
\quad \langle {\bf u}_0, {\bf Z} \rangle_{L^2} = 0,
\end{equation}
where ${\bf u}_0$ is the eigenfunction of $\ch$ for eigenvalue $\om_0$. Substituting (\ref{decomposition-last})
into (\ref{residual-equation}), we obtain the system of equations
\begin{equation}
\label{residual-equation-last} \left\{ \begin{array}{l}
i \dot{a} = \om_0 a + \langle {\bf u}_0, e^{-i \theta} {\bf F} \rangle_{L^2}, \\
i \dot{\bf Z} = \ch {\bf Z} + P_{a.c.}(\ch) e^{-i \theta} {\bf F}. \end{array} \right.
\end{equation}

We now set up our problem as an iteration scheme, where we look
for a fixed point in a small ball in a normed space. More
precisely, this space is composed of  all quadruples $(\om, \theta,
a,{\bf Z})$, equipped with  the norm
$$
\|(\om, \theta, a, {\bf Z}) \|_Z := \|\dot{\om}\|_{L^1_t} + \|\dot{\theta}-\om\|_{L^1_t} +
\| a \|_{L^2_t \cap L^{\infty}_t} + \|{\bf Z}\|_X.
$$
Note that the elements of the corresponding set are subject to the
appropriate initial conditions
$$
\om(0) \in (\om_0,\om_0+\epsilon), \quad \theta(0)=0, \quad a(0) = \langle {\bf u}_0, {\bf Y}(0) \rangle_{L^2}, \quad
{\bf Z}(0) = P_{a.c.}(\ch){\bf Y}(0).
$$

First, observe that
the matrix in front of the variables $\dot{\om}$ and
$\dot{\theta}-\om$ in \eqref{projection-equations} has the form
\begin{equation}\label{matrix-system}
\left[ \begin{array}{cc} \textup{Re} \langle\p_{\om} {\bf U}, {\bf
U}-{\bf U}_1 \rangle_{L^2}&  \textup{Im}
\langle\p_{\om} {\bf U}, {\bf U}_1\rangle_{L^2}\\
 \textup{Im} \langle\p_{\om}^2 {\bf U}, {\bf U}_1 \rangle_{L^2} &
 \textup{Re} \langle\p_{\om} {\bf U}, {\bf U}+{\bf
 U}_1 \rangle_{L^2} \end{array} \right] =
 \f{1}{2} \f{d}{d\om}\|{\bf U}\|_{L^2}^2 Id +O(\|{\bf U}_1\|_{L^2}).
\end{equation}
Due to the smallness of
$$
\|{\bf U}_1 \|_{L^{\infty}_t L^2_x} \leq \| a \|_{L^{\infty}_t} +
\|{\bf Z} \|_{L^{\infty}_t L^2_x}
$$
and the non-degeneracy condition \eqref{nondegeneracy}, we may conclude that the matrix 
(\ref{matrix-system}) is
invertible. (Note that $\|{\bf Z}\|_{L^{\infty}_t L^2_x}$ is a part of
the norm $\|{\bf Z}\|_X$, which is kept small in our fixed point
arguments.)

Next, we show that the quantities $\|\dot{\om}\|_{L^1_t}$ and
$\|\dot{\theta}-\om\|_{L^1_t}$ are under control. Indeed, due to
the invertibility of the matrix, the form of
\eqref{projection-equations}, and the quadratic nature of $F_1,
F_2$ (Proposition \ref{proposition-vector-field}), we have
\begin{eqnarray*}
\|\dot{\om}\|_{L^1_t} + \|\dot{\theta}-\om\|_{L^1_t} & \leq &
C(\|F_1\|_{L^1_t}+ \|F_2\|_{L^1_t}) \leq C \int_0^{\infty} \int_{\R}
|{\bf Y}(x,t)|^2 |{\bf U}(x)| dx dt \\
& \leq  &
C \|<x>^{2\al} {\bf U}\|_{L^1_x} \|<x>^{-\al} {\bf Y}\|_{L^\infty_x L^2_t}^2\leq
C \left( \| a\|^2_{L^2_t} +  \|{\bf Z}\|_X^2 \right).
\end{eqnarray*}
It follows from this bound that
\begin{eqnarray}
\label{local-norm-control}
\|\om - \om(0) \|_{L^{\infty}_t} + \|\theta - \int_0^t \om(s) ds \|_{L^{\infty}_t} \leq
C \left( \| a\|^2_{L^2_t} +  \|{\bf Z}\|_X^2 \right).
\end{eqnarray}

Since $\dot{\om} \in L^1_t$ and $\| \om - \om(0) \|_{L^{\infty}_t}$ is small, there exists $\epsilon_0 > \epsilon$ and
$\omega_{\infty} := \lim_{t \to \infty} \om(t)$ such that $\om_{\infty} \in (\om_0,\om_0+\epsilon_0)$
if $\om(0) \in (\om_0,\om_0 + \epsilon)$. Similarly there exists $\theta_{\infty} \in \R$ such that
$$
\lim_{t \to \infty} \left( \theta(t) - \int_0^t \omega(s) ds
\right) = \theta_{\infty}.
$$

Now, we control the quantity $\| a \|_{L^2_t \cap L^{\infty}_t}$. It follows
from the symplectic orthogonality conditions (\ref{symplectic-orthogonality}) that
$$
\langle {\bf u}_0, {\bf U}_1 \rangle_{L^2_x} = {\rm Re}
\langle {\bf u}_0 - \frac{{\bf U}}{\| {\bf U} \|_{L^2_x}}, {\bf U}_1 \rangle_{L^2_x} +
i {\rm Im}
\langle {\bf u}_0 - \frac{\partial_{\omega} {\bf U}}{\| \partial_{\omega} {\bf U} \|_{L^2_x}}, {\bf U}_1 \rangle_{L^2_x}.
$$
By Proposition \ref{proposition-local-bifurcation}, for any $\al \geq 0$, there is $C_{\al} > 0$ such that
$$
\left\| <x>^{\al} ( {\bf u}_0 - \frac{{\bf U}}{\| {\bf U} \|_{L^2_x}}) \right\|_{L^2_x} +
\left\| <x>^{\al} ( {\bf u}_0 - \frac{\partial_{\omega} {\bf U}}{\| \partial_{\omega} {\bf U} \|_{L^2_x}}) 
\right\|_{L^2_x} \leq C |\omega - \omega_0|.
$$
Therefore, we obtain
\begin{eqnarray*}
\| a \|_{L^2_t} & = & \| \langle {\bf u}_0, {\bf Y} \rangle_{L^2_x}\|_{L^2_t} =  \| \langle {\bf u}_0, {\bf U}_1 \rangle_{L^2_x}\|_{L^2_t} \\
& \leq & C \| \om - \om_0 \|_{L^{\infty}_t} \| <x>^{-\al} {\bf U}_1 \|_{L^{\infty}_x L^2_t} \leq
C (\epsilon + \|\om - \om(0) \|_{L^{\infty}_t}) \| <x>^{-\al} {\bf Y} \|_{L^{\infty}_x L^2_t},
\end{eqnarray*}
where $\epsilon + \| \omega - \om(0) \|_{L^{\infty}_t}$ is small due to smallness of $\epsilon$ and the bound (\ref{local-norm-control}).
Similarly, we obtain
\begin{eqnarray*}
\| a \|_{L^{\infty}_t} & = & \| \langle {\bf u}_0, {\bf Y} \rangle_{L^2_x}\|_{L^{\infty}_t} =
\| \langle {\bf u}_0, {\bf U}_1 \rangle_{L^2_x}\|_{L^{\infty}_t} \\
& \leq & C \| \om - \om_0 \|_{L^{\infty}_t} \| {\bf U}_1 \|_{L^{\infty}_t L^2_x} \leq
C (\epsilon + \|\om - \om(0) \|_{L^{\infty}_t}) \| {\bf Y} \|_{L^{\infty}_t L^2_x}.
\end{eqnarray*}

Finally, it remains to estimate the quantity $\|{\bf Z}\|_X$.
Due to our construction, we have ${\bf Z} = P_{a.c.}(\ch) {\bf Y}$, so that
we may apply the linear estimates \eqref{e:2}.
The nonlinearity $P_{a.c.}(\ch) e^{-i \theta} {\bf F}$ in the residual equation
\eqref{residual-equation-last} has two parts. The first part satisfies
\begin{eqnarray*}
\| P_{a.c.}(\ch) e^{-i \theta} (i\dot{\om}\p_\om {\bf U} + (\dot{\theta}-\om){\bf U})\|_{L^1_t H^1_x}
& \leq & C(\|\dot{\om}\|_{L^1_t} + \|\dot{\theta}-\om\|_{L^1_t})(\|{\bf
U}\|_{H^1_x}+\|\p_\om {\bf U}\|_{H^1_x}) \\
& \leq  & C \left( \| a\|^2_{L^2_t} +  \|{\bf Z}\|_X^2 \right).
\end{eqnarray*}
Roughly speaking, the second (nonlinear) part
$$
{\bf G} := P_{a.c.}(\ch) e^{-i \theta} \left( {\bf N}({\bf U} +
{\bf Y} e^{i \theta}) - {\bf N}({\bf U})\right),
$$
consists of the two terms
$$
{\bf G} \sim {\bf Y} {\bf U}^{2p}+ {\bf Y}^{2p+1},
$$
where ${\bf Y}$ is controlled in the $X$-norm by
$$
\| {\bf Y} \|_X \leq C (\| a \|_{L^2_t \cap L^{\infty}_t} + \|{\bf Z}\|_X).
$$

Note that
\begin{eqnarray*}
|{\bf G}(x,t)| + |\p_x {\bf G}(x,t) | \leq C (|{\bf
Y}| + |\p_x {\bf Y}|)(|{\bf U}|^{2p}+|\p_x {\bf U}|^{2p})+
C(|{\bf Y}|+|\p_x {\bf Y}|)|{\bf Y}|^{2p},
\end{eqnarray*}
We need to control the quantity $\|G\|_{\cn}$ in terms of $\|{\bf Y}\|_X$. We have
\begin{eqnarray*}
\|{\bf G}\|_{\cn} & \leq & C \|<x>^{\al} (|{\bf Y}| +
|\p_x  {\bf Y}|) (|{\bf U}|^{2p}+|\p_x {\bf U}|^{2p}) \|_{L^1_x L^2_t} +
C \| (|{\bf Y}|+ |\p_x {\bf Y}|)|{\bf Y}|^{2p}\|_{L^1_t L^2_x}\\
&\leq & C \left( \|\langle x \rangle^{-\al} {\bf Y} \|_{L^\infty_x L^2_t}+
\|\langle x\rangle^{-\al} \p_x {\bf Y} \|_{L^\infty_x L^2_t} \right)
\|<x>^{2\al} (|{\bf U}|^{2p} + |\p_x {\bf U}|^{2p})\|_{L^1_x L^{\infty}_t} \\
& \phantom{t} &  + C \|{\bf Y}\|_{L^\infty_t H^1_x} \|{\bf
Y}\|_{L^{2p}_t L^{\infty}_x}^{2p}.
\end{eqnarray*}

It is now easy to close the argument in the norm $\|{\bf Y}\|_X$.
Indeed, by Sobolev embedding for any $\epsilon > 0$
$$
\|{\bf Y}\|_{L^\infty_{t} L^\infty_x}\leq C\|{\bf Y}\|_{L^\infty_x
H^{1/2+\epsilon}_x}\leq C \|{\bf Y}\|_X.
$$
We also have $\|{\bf Y}\|_{L^4_{t} L^\infty_x}\leq \|{\bf Y}\|_X$
(by the definition of $\|\cdot\|_{X}$) and hence, for $p \geq 2$,
by the log convexity of the $L^{q}$ norms, we have
$$
\|{\bf Y}\|_{L^{2p}_t L^\infty_x}\leq \|{\bf Y}\|_{L^{4}_t
L^\infty_x}^{2/p} \|{\bf Y}\|_{L^{\infty}_t L^\infty_x}^{1-2/p}
\leq C \|{\bf Y}\|_X.
$$
All in all, combining the estimates for $\|{\bf G}\|_{\cn}$
with the estimates for $\|{\bf Y}\|_{L^{2p}_t L^\infty_x}$, we obtain
$$
\|{\bf G}\|_{\cn}\leq C
\|<x>^{2\al} (|{\bf U}|^{2p} + |\p_x {\bf U}|^{2p})\|_{L^1_x L^{\infty}_t}
\|{\bf Y}\|_X + C \|{\bf Y}\|_{X}^{2p+1}.
$$
By Proposition \ref{proposition-local-bifurcation}, there is $C > 0$ such that
$$
\|<x>^{2\al} (|{\bf U}|^{2p} + |\p_x {\bf U}|^{2p})\|_{L^1_x L^{\infty}_t} \leq C  \| \om - \om_0 \|_{L^{\infty}_t}
\leq C (\epsilon + \| \om - \om(0) \|_{L^{\infty}_t}).
$$
Since the last term is small due to the smallness of $\epsilon$ and the bound (\ref{local-norm-control}),
the fixed point argument is closed, and the proof of Theorem \ref{theorem-main} is complete.


\begin{thebibliography}{99}

\bibitem{Comech} G. Berkolaiko, A. Comech, ``On spectral stability of solitary waves 
of nonlinear Dirac equations on a line", preprint, arXiv:0910/0917.

\bibitem{Boussaid}  N. Boussaid, ``Stable directions for small nonlinear Dirac standing waves," {\em  Comm. Math. Phys.} {\bf  268} (2006),  no. 3, 757--817.

\bibitem{BS} V. Buslaev, C. Sulem, ``On the stability of solitary waves for Nonlinear Schr\"{o}odinger
equations", Annales Institut Henri Poincar\'e, Analyse Nonlineaire {\bf 202} (2003), 419--475.

\bibitem{ChPel-cme} M. Chugunova, D. Pelinovsky,  ``Block-diagonalization of the symmetric first-order coupled-mode system," {\em SIAM J. Appl. Dyn. Syst.} {\bf 5} (2006), 66--83.

\bibitem{Cuc} S. Cuccagna, ``On asymptotic stability in energy space of ground states of NLS in 1D",
{\em J. Diff. Eqs.} {\bf 245} (2008), 653--691.

\bibitem{CM} S. Cuccagna, T. Mizumachi, ``On asymptotic stability in energy space of ground states for nonlinear 
Schr\"odinger equations", {\em Comm. Math. Phys.} {\bf 284} (2008), 51--77.

\bibitem{CucTr} S. Cuccagna, M. Tarulli, ``On asymptotic stability of standing waves
of discrete Schr\"{o}dinger equation in $\mathbb{Z}$", SIAM J.
Math. Anal. {\bf 41} (2009), 861--885.

\bibitem{Gang} Z. Gang, I.M. Sigal, ``Asymptotic stability of nonlinear Schr\"{o}dinger equations
with potential'', {\em Rev. Math. Phys.} {\bf 17} (2005), 1143--1207.

\bibitem{GS} M. Goldberg and W. Schlag, ``Dispersive estimates for Schr\"{o}dinger
operators in dimensions one and three'', {\em  Commun. Math. Phys.} {\bf  251}, 157--178 (2004)

\bibitem{GWH} R.H. Goodman, M.I. Weinstein, and P.J. Holmes, ``Nonlinear propagation
of light in one-dimensional periodic structures'', J. Nonlinear.
Science {\bf 11} (2001), 123--168.

\bibitem{GSWK} R.H. Goodman, R.E. Slusher, M.I. Weinstein, and M.
Klaus, ``Trapping light with grating defects", {\em Contemp. Math.} {\bf
379} (2005), 83--92.

\bibitem{KPS} P. Kevrekidis, D. Pelinovsky, A. Stefanov, "Asymptotic stability
of small solitons in the discrete nonlinear Schr\"{o}dinger
equation in one dimension", {\em SIAM J. Math. Anal.}, {\bf 41}
(2009), 2010-2030.

\bibitem{KirrZ} E. Kirr, A. Zarnescu, ``Asymptotic stability of ground states
in 2D nonlinear Schr\"{o}dinger equation including subcritical
cases", {\em J. Diff. Eqs.} {\bf 247} (2009), 710--735.

\bibitem{KirrM} E. Kirr,  \"{O}. Mizrak, ``Asymptotic stability of ground states
in 3D nonlinear Schr\"{o}dinger equation including subcritical
cases", {\em J. Funct. Anal.} {\bf 257} (2009), 3691--3747.

\bibitem{Machihara} S. Machihara, M. Nakamura, K. Nakanishi, T. Ozawa,  ``Endpoint Strichartz estimates and
global solutions for the nonlinear Dirac equation", {\em J. Funct. Anal.} {\bf 219} (2005), 1--20.

\bibitem{Machihara2} S. Machihara, K. Nakanishi, T. Ozawa,
``Small global solutions and the nonrelativistic limit for the
nonlinear Dirac equation", {\em Rev. Math. IberoAm.} {\bf 19}
(2003), 179--194.

\bibitem{Mizum} T. Mizumachi,  ``Asymptotic stability of small solitary waves
to 1D nonlinear Schrödinger equations with potential", {\em  J.
Math. Kyoto Univ. } {\bf 48}  (2008), 471--497.

\bibitem{M2} T. Mizumachi,  ``Asymptotic stability of small solitons
for 2D nonlinear Schrdinger equations with potential", {\em J.
Math. Kyoto Univ.} {\bf 47} (2007), 599--620.

\bibitem{NO} M. Nakamura, T. Ozawa,  ``The Cauchy problem for
nonlinear Klein-Gordon equations in the Sobolev spaces", {\em
Publ. Res. Inst. Math. Sci.} {\bf 37} (2001), 255--293.

\bibitem{PW} C.A. Pillet, C.E. Wayne, ``Invariant manifolds
for a class of dispersive, Hamiltonian, partial differential
equations'', {\em J. Diff. Eqs.} {\bf 141} (1997), 310--326.

\bibitem{ChPorPel} M.A. Porter, M. Chugunova, D.E. Pelinovsky,
``Feshbach resonance management of Bose--Einstein condensates in
optical lattices'', Phys. Rev. E {\bf 74} (2006), 036610-8.

\bibitem{RS} I. Rodnianski, W. Schlag, \emph{Time decay for solutions of Schrödinger equations with rough and time-dependent potentials.}, {\em  Invent. Math.} {\bf 155}   (2004),  no. 3, 451--513.

\bibitem{SS} H. Smith, C. Sogge, \emph{Global Strichartz estimates for nontrapping perturbations of the Laplacian}, {\em Comm. Partial Differential Equations} {\bf 25}  (2000),  no. 11-12, 2171--2183.

\bibitem{SW} A. Soffer, M.I. Weinstein, ``Selection of the ground state for nonlinear
Schr\"{o}dinger equations'', {\em Rev. Math. Phys.} {\bf 16} (2004), 977--1071.

\bibitem{Weder} R. Weder, \emph{The $W_{k,p}$-continuity of the Schr\"{o}dinger wave operators
on the line}, {\em Comm. Math. Phys.}  {\bf 208}, 507--520 (1999).

\bibitem{Weder2} R. Weder, \emph{$L^p$--$L^{p'}$ estimates for the Schr\"{o}dinger
equation on the line and inverse scattering for the nonlinear
Schr\"{o}dinger equation with a potential},{\em  J. Funct. Anal. } {\bf
170}, 37--68 (1999).

\bibitem{YT} H.T. Yau, T.P. Tsai,  ``Asymptotic dynamics of nonlinear Schr\"{o}dinger
equations: resonance dominated and radiation dominated solutions'', {\em Comm. Pure Appl. Math.}
{\bf 55} (2002), 1--64.

\end{thebibliography}
\end{document}